\documentclass[aps,prx,showpacs,twocolumn,reprint,superscriptaddress]{revtex4-2}

\usepackage{qcircuit}
\usepackage[dvips]{graphicx}
\usepackage{amsmath,amssymb,amsthm,mathrsfs,amsfonts,dsfont,mathtools}
\usepackage{epsfig}
\usepackage{braket}
\usepackage{hyperref}
\usepackage{bm}
\usepackage{enumerate}
\usepackage{color}
\usepackage{graphicx}
\usepackage{pgf}
\usepackage{pgfplots}
\pgfplotsset{compat=1.17}
\usepackage{tikz}
\usetikzlibrary{pgfplots.groupplots}
\usetikzlibrary{external}
\tikzsetexternalprefix{figures/}
\tikzset{external/force remake}
\usepackage{enumitem}
\usepackage[capitalize]{cleveref}
\usepackage[normalem]{ulem}

\usepackage{amsthm}

\newtheorem{statement}{Statement}
\newtheorem{theorem}{Theorem}

\Crefname{theorem}{Theorem}{Theorems}

\newcommand{\tr}{\mathrm{Tr}}

\newcommand{\var}{\mathrm{Var}}

\newcommand{\dmat}{\mathbf{D}}
\newcommand{\cmat}{\mathbf{C}}
\newcommand{\ttot}{N_T}
\newcommand{\nshot}{N_s}

\newcommand{\nc}{c}
\newcommand{\xmat}{\mathbf{X}}

\newcommand{\nobs}{N_o}
\newcommand{\rhoerr}{\rho_{err}}
\newcommand{\rhoid}{\rho_{id}}

\newcommand{\qmaddress}{Quantum Motion, 9 Sterling Way, London N7 9HJ, United Kingdom}

\newcommand{\mathaddress}{\affiliation{Mathematical Institute, University of Oxford, Woodstock Road, Oxford OX2 6GG, United Kingdom}}

\begin{document}

\title{Algorithmic Shadow Spectroscopy}

\author{Hans Hon Sang Chan}
\email{hans.chan@materials.ox.ac.uk}
\affiliation{Department of Materials, University of Oxford, Parks Road, Oxford OX1 3PH, United Kingdom}
\author{Richard Meister}
\affiliation{Department of Materials, University of Oxford, Parks Road, Oxford OX1 3PH, United Kingdom}
\author{Matthew L. Goh}
\affiliation{Department of Materials, University of Oxford, Parks Road, Oxford OX1 3PH, United Kingdom}
\mathaddress
\author{B\'alint Koczor}
\email{balint.koczor@materials.ox.ac.uk}
\affiliation{Department of Materials, University of Oxford, Parks Road, Oxford OX1 3PH, United Kingdom}
\affiliation{\qmaddress}
\mathaddress

\begin{abstract}
We present shadow spectroscopy as a simulator-agnostic quantum algorithm for estimating energy gaps using very few circuit repetitions (shots) and no extra resources (ancilla qubits) beyond performing time evolution and measurements. The approach builds on the fundamental feature that every observable property of a quantum system must evolve according to the same harmonic components: we can reveal them by post-processing classical shadows of time-evolved quantum states to extract a large number of time-periodic signals $N_o\propto 10^8$, whose frequencies correspond to Hamiltonian energy differences with Heisenberg-limited precision. We provide strong analytical guarantees that (a) quantum resources scale as $O(\log N_o)$, while the classical computational complexity is linear $O(\nobs)$, (b) the signal-to-noise ratio increases with the number of processed signals as $\propto \sqrt{\nobs}$, and (c) spectral peak positions are immune to reasonable levels of noise. We demonstrate our approach on model spin systems and the excited state conical intersection of molecular CH$_2$ and verify that our method is indeed intuitively easy to use in practice, robust against gate noise, amiable to a new type of algorithmic-error mitigation technique, and uses orders of magnitude fewer number of shots than typical near-term quantum algorithms -- as low as 10 shots per timestep is sufficient. Finally, we measured a high-quality, experimental shadow spectrum of a spin chain on readily-available IBM quantum computers, achieving the same precision as in noise-free simulations without using any advanced error mitigation, and verified scalability in tensor-network simulations of up to 100-qubit systems.
\end{abstract}

\maketitle

\section{Introduction}
Quantum simulation is possibly the most natural use of quantum computers, with numerous potential applications in, e.g., understanding
quantum field theory~\cite{latticeschwinger}, quantum gravity~\cite{jafferis2022traversable}, and may help us develop novel drugs and materials~\cite{cao2019quantum, mcardle2020quantum,bauer2020quantum, Motta2022}.
The hope is that future quantum hardware developments will lead to scalable, universal quantum computers
that can simulate the time-evolution of any quantum system. Many such Hamiltonian simulation
algorithms are known,
such as product formula or Trotterisation~\cite{berry2007efficient}, quantum signal processing~\cite{Low2017optimal}, qubitization~\cite{Low2019hamiltonian} or linear combination of unitaries~\cite{Berry2015} to name a few.
The quantum circuit depth of these algorithms is only linear in the total simulation time, in stark contrast with the exponential depth of generic classical simulation
algorithms known today -- thus promising a \textit{quantum advantage} in modelling systems of practical interest.

Unfortunately, highly non-trivial measures for quantum error correction will be necessary to run these Hamiltonian simulation circuits, and 
the technological overhead associated with even the simplest error correction schemes means they
remain prohibitive in the current era of noisy intermediate scale quantum (NISQ) technology -- a single logical qubit, for example, may need to be embedded into
thousands of physical, noisy qubits.
In response, many researchers are proposing sophisticated error mitigation
techniques and new quantum simulation algorithms~\cite{kim2023evidence, error_mitigation,PhysRevX.11.031057,PhysRevX.12.041022,farhi2014quantum,peruzzo2014variational,endoHybridQuantumClassicalAlgorithms2021, cerezoVariationalQuantumAlgorithms2021a, bharti2021noisy} that seek to demonstrate practical utility with near-term quantum computers. In this paradigm, a parametrised shallow
quantum circuit is typically used, and its parameters are updated by a classical computer
such that the evolution of the parametrised quantum state closely approximates the
trajectory of the true quantum evolution~\cite{ying_li,PhysRevLett.125.180501}.
Alternatively, analog quantum simulators built from e.g. ultracold-atoms can mimic the dynamics of model 
systems, including variants of the continually elusive Fermi-Hubbard model~\cite{tarruell2018ultracold,Hart2015,Schneider2012}.

\begin{figure*}[t]
	\begin{centering}
		\includegraphics[width=.95\linewidth]{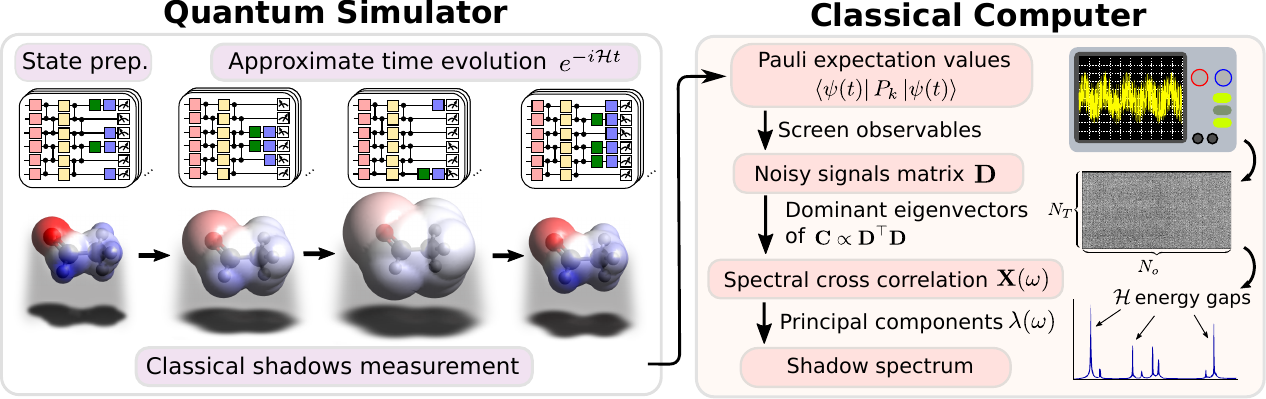}
	\end{centering}
	\caption{ 
		Flowchart of algorithmic shadow spectroscopy: a plethora of well-established quantum simulation algorithms can be used to generate approximate time-evolved quantum states $|\psi (t_n) \rangle$---in this example a mean-field many-electron wavefunction of the acetaldehyde molecule---from which we extract a series of classical shadows using randomised measurements. The shadows are stored classically as a list of binary numbers and measurement basis indexes. These are post-processed using conventional HPC resources by computing a large
		noisy signal matrix $\textbf{D}$, then the square matrix $\textbf{C}$,
		and finally the shadow spectrum; these steps are elaborated upon in \cref{fig:fig1}. 
		\label{fig:summary}
	}
\end{figure*}

And yet, even the ability to approximate time-evolution with lower quantum resources is itself not useful enough in practice, because quantum mechanics poses fundamental limitations on how efficiently we can extract observable properties from a quantum simulation.
This constraint is exemplified by the `holy grail' of quantum simulation; the computation of a Hamiltonian's energy spectrum, a task that has direct experimental relevance in e.g. predicting electronic properties of solids to photo-chemical processes~\cite{singleton2001band, gonzalez2021quantum}, but its accurate classical treatment generally suffers from prohibitive scaling.
Powerful quantum phase estimation protocols involving Hamiltonian simulation, for example, are designed for such a task, but require extra ancillary qubit registers controlling the already challenging time-evolution circuits. The probability of a successful phase estimation measurement also imposes strict demands on the overlap between initial states and the target eigenstates~\cite{garnetchan2022}.
In NISQ-friendly ansatz-based methods that avoid additional qubits like the variational quantum eigensolver (VQE), the measurement challenge instead manifests as the need for many circuit measurement repeats (shots) for the classical optimisation of the fixed quantum circuit ansatz.

In this work, we propose and develop a novel means of estimating the energy spectrum differences, or gaps, in a Hamiltonian which we call `shadow spectroscopy'. It requires no additional quantum resources such as ancilla qubits, and can in principle work on any quantum 
simulator platform (digital or analogue) as well as any Hamiltonian simulation method of choice. The initial state only needs to have sufficient support for the set of eigenstates which we want to know the energy gaps for, instead of a high overlap with one target eigenstate.

One key innovation is to harness powerful classical shadow techniques~\cite{classical_shadows, shadows_derandomization,wanMatchgateShadowsFermionic2022, zhaoFermionicPartialTomography2021,akhtarScalableFlexibleClassical2022, bertoniShallowShadowsExpectation2022} which enables us to capture a large number of observables of a quantum state despite using orders of magnitude fewer measurement shots than e.g. VQE. Our strategy is to estimate classical shadows of time-evolved quantum states, and analyse how that information---as captured by
the shadows---evolves in time. Due to the fundamental laws of quantum mechanics, oscillations of all observable properties in a quantum system contain exactly the same fundamental frequencies that precisely correspond to energy differences
in the Hamiltonian. These can then be extracted via classical post-processing techniques -- illustrated in \cref{fig:summary}.

We tested our technique experimentally and numerically in a broad range of applications and in various model systems of up to 100 qubits, finding it very easy to use in practice and provides a promising 
alternative for computing excitation energies.  
We also confirm that the present approach is robust to gate errors
due to our effective reconstruction of a large number of observables from randomised measurements, and that positions of the ideal spectral peaks are not altered by incoherent errors;
indeed, artefacts in the spectrum introduced by gate errors can even decrease as we increase the system size.
Our approach is also amiable to a new type of quantum error mitigation technique~\cite{PhysRevA.99.012334} that mitigates algorithmic errors incurred during the approximate time-evolution. 
Another significant advantage is its robustness to finite circuit repetitions (shot noise) as illustrated in \cref{fig:fig1}.
In nearly all examples explored, only a few dozen measurements of the state were needed to resolve dominant peaks in the spectrum, whereas figures well beyond $10^6$ shots are common for VQE. Specifically, we show that estimating and analysing $\nobs$ signals only requires $\log \nobs$ shots and a linear classical post-processing time;
having access to a large $\nobs$, however, allows us to efficiently find `useful observables', to improve spectral signal-to-noise ratio,
and to suppress the effect of gate noise. In particular, post-processing $\nobs = 10^8$ observables
can be performed in less than an hour with HPC resources and results in an improvement of the spectral intensity of $\propto 10^{4}$.
Finally, our conclusions on the method's robustness against gate errors and shot
	noise are reinforced by the experimental demonstration of a shadow spectrum measured on
	freely-accessible noisy IBM quantum computers -- finding excellent agreement between numerical
	simulation and actual hardware results without using any advanced error mitigation techniques.

The structure of this work is as follows. We first introduce the core ideas, and then discuss technical implementations in \cref{sec:shadow_spec}. We then illustrate the utility of our approach on a broad set of examples in \cref{sec:application}. We finally discuss related questions in \cref{sec:discussion} and conclude in \cref{sec:conclusion}. While the method is robust and versatile, we discuss in detail its main limitations in \cref{app:limitations}. We also report related techniques known in the literature in \cref{app:lit}.

\section{Shadow spectroscopy~\label{sec:shadow_spec}}
In optical spectroscopy, an observable property of a quantum state, usually the expected value of some transition dipole moment (e.g. the magnetic dipole in Nuclear Magnetic Resonance),
is recorded as a function of time, and the Fourier transformed signal reveals the desired transition energies. One could understand shadow spectroscopy as `substituting' measurement of time-dependent expectation values with classical shadows computed on a quantum device. As we demonstrate in numerical and hardware results, it offers great flexibility, and determining energy differences of low-lying single-electron excitation using only local Pauli shadows and classical post-processing is straightforward. In this section, we present the main technical implementation details.

\subsection{Preliminaries}
Time-evolution of any quantum state under Hamiltonian dynamics
introduces periodic oscillations precisely determined by differences of eigenvalues in the Hamiltonian operator.
The following statement summarises the well-known property that expected values of observables 
oscillate with these frequencies under a closed Hamiltonian evolution~\cite{gu2022noise,gnatenko2022energy}.  
\begin{statement}\label{statement:intenisty}
Given any $d$-dimensional Hamiltonian $\mathcal{H}$ with eigenvectors $\ket{\psi_k}$,
one can define any quantum state as
$| \psi \rangle = \sum_{k = 1}^d c_k \ket{ \psi_k }$.
Measuring the expected value of any observable $O$ 
in the time-evolved quantum state $\ket{\psi(t)} := e^{- i t \mathcal{H}} \ket{\psi}$ yields
the signal $ S(t)$ as 
\begin{equation}
	S(t)  := \langle \psi(t) |  O |\psi(t) \rangle = \sum_{k,l = 1}^d  I_{kl} \, e^{-i t (E_l - E_k) }.
\end{equation}
The above Fourier components have frequencies that correspond to differences of eigenvalues $E_l - E_k$
of the Hamiltonian. The intensity of the corresponding peak in the Fourier spectrum is $I_{kl} = c_k^* c_l   \langle \psi_k | O | \psi_l \rangle$.
\end{statement}
\begin{proof}
	All $d$-dimensional Hamiltonians admit the spectral decomposition
	$    \mathcal{H} = \sum_{k=1}^{d} E_k |\psi_k \rangle\langle \psi_k|$
	with eigenvectors $|\psi_k \rangle$ spanning the full space.
	Thus, the evolution operator acts as
	\begin{equation*}
		| \psi(t) \rangle =e^{-i t \mathcal{H} }  | \psi \rangle = \sum_{k = 1}^q c_k    e^{-i t E_k } | \psi_k \rangle,
	\end{equation*}
	and the time-dependent $\langle \psi(t) |  O |\psi(t) \rangle$ expected value is then obtained via a direct calculation.
	The generalisation of this statement to density matrices as well as to infinite-dimensional but bounded
	Hamiltonians follows trivially.
\end{proof}

An important consequence of the above statement is that time-dependent expected values contain the same harmonic components $e^{-i t (E_l - E_k) }$ regardless of the
observable measured. These frequencies precisely correspond to energy differences as $(E_l - E_k)$ in a Hamiltonian
which---in principle---can be determined by naively picking a single observable $O$ and measuring its corresponding time-dependent expectation value signal $S(t)$,
then Fourier transforming $S(t)$.

In practice, a number of crucial factors determine the intensity of $S(t)$.
First, it depends on the amplitudes of the statevector $c_k$; if the initial state is, e.g., a random state or a uniform superposition over all eigenstates, the Fourier spectrum may contain an exponentially large number of overlapping peaks and may thus prohibit us from learning anything from the spectrum. If we can instead initialise the quantum state such that only
a handful of lower-energy eigenstates contribute, i.e.,  $|c_k| \gg |c_l|$ where all $k \in \mathcal{Q}$ and $l \notin \mathcal{Q}$,
then we can isolate $\mathcal{O}( |\mathcal{Q}|^2 )$ dominant peaks. This is the case, for example, when we can 
initialise
in a superposition of only a few eigenstates as $ \ket{ \psi}  = \sum_{k \in \mathcal{Q}} c_k \ket{\psi_k}$.
In fact, we consider several known strategies in this work for preparing initial states
such that $|c_1| \approx 1$ and $|c_k| \gg |c_l|$ for all $2 \leq k \leq  q$ and $l>q$. Then we only have $q$ dominant peaks of
non-zero frequency from which we can learn the lowest $q$ excitation energies as $E_2 - E_1$, $E_3 - E_1$ etc.

Second, the signal intensity depends on the matrix elements between eigenstates as $\langle \psi_k | O | \psi_l \rangle$ which can be arbitrarily
small for any chosen observable $O$.
One can identify problem specific observables that yield dominant signals, such as excitation operators in quantum chemistry
as we discuss in \cref{sec:qchem}; in \cref{sec:quantumexcite} we also discuss how techniques from subspace expansion are similarly relevant~\cite{mccleanSubspaceExpansion}.
However, even if we can guarantee that a certain observable yields a non-zero matrix element,
in practice it may still require a large number of measurements to sufficiently suppress shot noise, especially when a signal has lower intensity than the shot noise background.

For this reason we employ classical shadows, which is designed to reconstruct a large set of observable expectation 
values and thus multiple $S_i(t)$. This allows us to search efficiently for local observables that yield intense signals.
We discuss several other advantages of estimating many observables in \cref{sec:discussion} e.g. it can boost the signal-to-noise ratio, thereby overcoming limitations of resolving relatively low intensity signals using only a logarithmic number of shots
-- and it also makes the approach noise robust.

\subsection{Reconstructing many time-dependent signals using classical shadows}

Classical shadow procedures~\cite{classical_shadows} apply randomised measurements
to several copies of an unknown quantum state, and enable us to estimate a very large number of properties with provable sample complexities. 
While essentially all variants of classical shadows are useful for our purposes~\cite{shadows_derandomization,wanMatchgateShadowsFermionic2022, zhaoFermionicPartialTomography2021,akhtarScalableFlexibleClassical2022, bertoniShallowShadowsExpectation2022}, which we discuss below in \cref{sec:discussion}, in this section
we use specifically the near-term friendly variant that enables us to reconstruct expected values of local Pauli operators.
Let us briefly recapitulate the main steps of the approach.

\begin{itemize}[leftmargin=*]
	\item  We apply a random unitary $U$ to rotate a copy of an unknown quantum state. For local Pauli strings, the unitaries are chosen randomly from single qubit Clifford gates on each qubit and the procedure is thus equivalent to randomly selecting to measure in the $X, Y$ or $Z$ bases -- we measure each qubit to obtain $N$-bit measurement outcomes $\ket{\hat{b}_i} \in \{0,1\}^N$.
	\item We then generate the classical shadows by applying the inverse of the measurement channel $\mathcal{M}$, which can be done efficiently as the channel chosen is a distribution over Clifford circuits. The classical snapshots are generated as $\hat{\rho}_i = \mathcal{M}^{-1} ( U^\dagger \ket{\hat{b}_i}\bra{\hat{b}_i} U ) $, the classical shadows are collections of these $N_s$ snapshots $S(\rho; N_s) = [\hat{\rho}_1,...,\hat{\rho}_{N_s}]$.
	\item From these classical shadows we can construct $K$ estimators of $\rho$ from our $N_s$ snapshots as \\ $\hat{\rho}_{(k)}=\frac{1}{r} \sum_{i=(k-1)r+1}^{k r} \hat{\rho}_i$ with $r=\lfloor N_s/K \rfloor$ and classically calculate estimators of the Pauli expectation values $\hat{o}_i(N_s,K) = \text{median} \{ \tr(O_i \hat{\rho}_{(1)}),...,\tr(O_i \hat{\rho}_{(K)}) \}$.
	\item The sample complexity of obtaining these estimators of $N_o$ Pauli operators of locality $q$ to error $\epsilon$ is $\mathcal{O} [3^q \log(N_o )/\epsilon^2]$. 
\end{itemize}

In this context, the classical shadow approach is very NISQ friendly as it only requires applying random single-qubit rotations immediately prior to measurement in the
standard basis.
The approach also has a provable sample efficiency: Recall that $q$-local Pauli strings are tensor products of single-qubit Pauli operators $P^{(q)} \in \{\mathrm{Id}, X, Y, Z \}^{\otimes N}$ such that $P^{(q)}$ acts non-trivailly (with $X$, $Y$ or $Z$) to $q$ qubits and trivially (with $\mathrm{Id}$) on all other qubits.
To reconstruct $\nobs$ at most $q$-local Pauli
strings one needs to collect a number of samples $N_s = \mathcal{O}(\log \nobs)$ that grows only logarithmically as we increase $\nobs$.
Furthermore, the classical computational resources are linear as $\mathcal{O}(\nobs)$ and fast, optimised, and hardware agnostic code can e.g. be found in the QuEST~\cite{QuESTlink} family
of quantum emulation software.
By applying this procedure, we can determine $\nobs$ time-dependent signals with a similarly advantageous sample complexity.

Let us briefly summarise the sample complexity of our approach of estimating
time-dependent signals using classical shadows.
\begin{statement}\label{statement:shadow}
	Determining $\nobs$ independent signals of $\ttot$ time increments as time-dependent expected values
	of at most $q$-local Pauli strings requires a number of circuit repetitions as
	\begin{equation}
	 N_s \in \mathcal{O} [ \ttot 3^q  \epsilon^{-2} \log(\nobs) ],
\end{equation}
using the classical shadows technique.
\end{statement}
\begin{proof}
	Following Theorem~1 of Ref.~\cite{classical_shadows}, $\nobs$ Pauli strings $O_i$ of locality $q$,
	can be estimated to precision parameters $\epsilon, \delta$ via the number
	of batches $K= 2\log(2N_o/\delta)$ where the number of samples in the individual batches is
	$N_{batch}=\tfrac{34}{\epsilon^2} \max_{i} \lVert O_i \rVert_{shadow}^2$.
	This results in an overall number of samples $N_s = N_{batch} K$ while the norm is given in Lemma~3 in Ref.~\cite{classical_shadows}
	as $\lVert O_k \rVert_{shadow}^2 = 3^q$.
	As such, at each time increment we use this procedure once to estimate $\nobs$ Pauli strings
	of locality $q$ which has a sample complexity upper bounded by
	\begin{equation}
		\tfrac{68}{\epsilon^2} 3^q \log(2\nobs/\delta).
	\end{equation}
	The procedure is repeated for each time increment and is thus used overall $\ttot$ times from
	which the total sample complexity follows.
\end{proof}

The above statement guarantees us that we can obtain a large number $\noindent$ of signals with logarithmic efficiency.
However each signal is still burdened with shot noise $S_i(t) + \mathcal{E}_i(t)$ such that the variance of the random
fluctuations $|\mathcal{E}_i(t)| \leq \epsilon$ is globally bounded by the precision parameter $\epsilon$ inherited
from classical shadows. One might feel tempted to average the signals as $\sum_i S_i(t)$.
However, as we show, each signal is phase shifted by an amount that 
is specific to each observable as
\begin{equation}\label{eq:evol}
    S_i(t) = 
    \mathrm{const} + \sum_{\substack{k,l=1\\k <l}}^q |I_{ikl}| \cos[ t (E_l - E_k) + \phi_{ikl}].
\end{equation}
Here we have used that $I_{ikl}= |I_{ikl}| e^{- i \phi_{ikl}} = c_k^* c_l  \langle k | O_i | l \rangle$.
Since averaging over phase-shifted cosines may even cancel out the signal, and that individual phase shifts are not known in advance, 
our problem ultimately becomes that of finding the most intense common frequency components across a collection of noisy signals that share common harmonics.

\begin{figure*}[t]
	\begin{centering}
		\includegraphics[width=\linewidth]{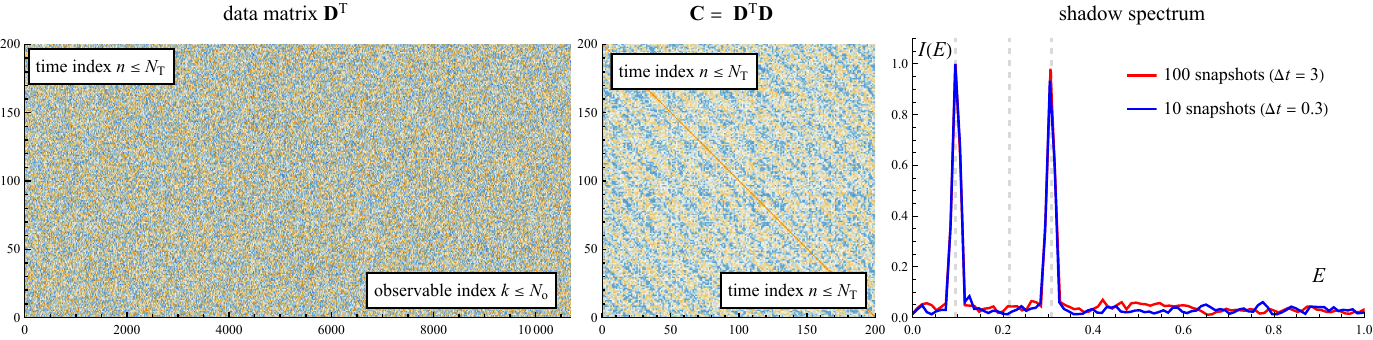}
	\end{centering}
	\caption{ 
		  A $14$-qubit initial state $|\psi(0)\rangle \propto (1, \tfrac{1}{10}, \tfrac{1}{10}, 0, 0 \dots)$
		as a superposition of the 3 lowest lying energy states of a spin-problem similar to the one in \cref{sec:spin}).
		(left) At each time increment ($\ttot = 200$)
		we estimate every up to 3-local Pauli string ($\nobs = 10689$)
		from classical shadows of $\nshot = 100$ snapshots. Overall $2000/10689$ signals
            are distinguishable from pure shot noise.
		(middle) Time-periodic correlations are clearly visible in $\cmat \propto \dmat^T \dmat$.
		(right)
		The shadow spectrum is obtained by effectively Fourier transforming the dominant eigenvectors
		of $\cmat$ (see details in main text).
		Peaks in the spectrum of the dominant eigenvectors of $\cmat$ reveal energy differences in the Hamiltonian
		as $E_0-E_1$ and $E_0-E_2$  (dashed lines are exact energies).
		A third peak does not appear as it is suppressed due to the initial state $\propto 0.1^2$, see text.
		Blue line: the signal-to-noise ratio is unchanged when reducing the number of snapshots to $\nshot = 10$
            and proportionally increasing the timesteps $\ttot = 2000$.
		The signal-to-noise ratio is increased as $\propto \sqrt{\nobs}$ as we increase the number of
		observables $\nobs$ while the classical computational cost is linear $\mathcal{O}(\nobs)$
		and the required quantum resources ($\nshot$) are only logarithmic $\mathcal{O}(\log \nobs)$.
		The precision of energy difference estimation exhibits Heisenberg scaling as it is proportional
            to the total simulation time (linear in the used quantum resources).
		\label{fig:fig1}
	}
\end{figure*}

\subsection{Classical post-processing\label{sec:postproc}}
We first standardise the signals using the sample mean $\mu_k $ and the empirical standard deviation  $\sigma_k$ as
\begin{equation} \label{eq:standardised_signals}
	f_k(n) = \frac{ \langle \psi(t_n) | P_k |\psi(t_n) \rangle - \mu_k } { \sigma_k }.
\end{equation}
Above we have also introduced the discrete temporal index $n \leq \ttot$ and the time variable $t_n = n \Delta t$.
We first pre-screen the data to filter out signals $f_k(n)$ that are not sufficiently different from statistical noise
-- an autocorrelation test, such as the Ljung-Box test \cite{box1970distribution,ljung1978measure}, works very well
in practice, requires linear computational time $\mathcal{O}(\nobs)$ and can be completely parallelised requiring no communication
between the classical processors, e.g., 
every CPU in a cluster receives a copy of the shadow data (a relatively small dataset) and processes it independently.
We then define a data matrix $\dmat \in \mathds{R}^{\nobs \times \ttot}$ that contains these
signals $[\dmat]_{kn} = f_k(n)$ as row vectors and 
assume the column dimension is much larger than the row dimension $\nobs \gg \ttot$ 
due to our effective use of classical shadows as illustrated in \cref{fig:fig1}.

We then perform a dimensionality reduction in linear time as $\mathcal{O}(\nobs \ttot^2)$.
This involves computing a relatively small
square matrix $\cmat \in \mathds{R}^{\ttot \times \ttot}$ and finding its $\nc$ dominant eigenvectors
 $v_{1}, v_{2}, \dots  v_{\nc}$ that maximise the average overlap with our experimentally estimated signals $f_k$.
For example, the dominant eigenvector satisfies
\begin{equation} \label{eq:maximal_overlap}
	v_1^T \cmat v_1 =
	\max_{\lVert v\rVert_2=1}
	\frac{1}{\nobs} \sum_{k=1}^{N_o} |\langle f_k, v \rangle|^2.
\end{equation}
Given $\nc$ peaks in the spectrum, the matrix $\cmat$ has $2\nc$ dominant eigenvectors which form the signal subspace (as in subspace
methods in signal processing \cite{hayes1996statistical}), see further details in \cref{app:dimred}.
The shadow spectrum is effectively obtained by Fourier transforming the dominant eigenvectors $v_{1}, v_{2}, \dots  v_{\nc}$
	that form the signal subspace.
Specifically, we robustly estimate the spectral density among these vectors using further classical processing:
We calculate the dominant singular
value of a spectral cross-correlation matrix in time $\mathcal{O}(\nc^3 \ttot)$ that is independent of $\nobs$
as we detail in \cref{app:cross_corr}.

\subsection{Robustness to shot noise}

We build analytical models in \cref{app:shot_noise} to understand how statistical shot noise due to
finite circuit repetition, using $\nshot$ shots, is suppressed as we increase the number of observables:
We find that the signal-to-noise ratio of a peak is proportional to
$\propto \ttot \nshot \overline{I^2} \sqrt{\nobs}$ 
where  $\overline{I^2} :=  \tfrac{1}{\nobs} \sum_{i=1}^{\nobs}   I_i^2$
is the average intensity in the group of $\nobs$ observables.
\cref{fig:fig1} (right) nicely illustrates that if we decrease the number of shots
from $100$ to $10$ per timestep but proportionally increase the number of timesteps,
the spectrum is qualitatively unchanged since the product 
$\ttot \nshot $ in the above expression for the signal-to-noise ratio is unchanged.

We can contrast this with a naive approach whereby one Pauli observable is picked randomly and
we estimate its signal (i.e. randomly pick a single row of $\dmat$);  by computing its spectrum
one would obtain a peak intensity on average $ \ttot \nshot \overline{I^2} $.
Thus, shadow spectroscopy boosts the signal intensity by a factor $\sqrt{\nobs}$.
For example, we estimate that post-processing $\nobs = 10^8$ observables can be performed in less
than an hour on a desktop PC 
and boosts the intensity by a factor $10^4$ using
only a logarithmic overhead in quantum resources.

We perform numerical simulation using both Matrix Product State (MPS) simulations (up to 100 qubits)
and exact time-evolution (up to 14 qubits)  in \cref{app:scaling} and confirm that the performance
of our approach indeed improves as we increase the number of qubits. The reason is that an increasing system
size gives rise to an increasing number of useful observable signals that our post-processing approach can harness
-- hence in our MPS simulations we see an improved signal-to-noise ratio at a system of 100 qubits.

We now summarise the algorithm in \cref{fig:summary} and state its computational complexity
-- refer to \cref{app:post_proc} for further details.
\begin{statement}\label{statement:eigenvec}
	We first compute the data matrix $\dmat \in \mathds{R}^{\nobs \times \ttot}$ of estimated signals with
	$\nobs \gg \ttot$ as defined via \cref{eq:standardised_signals}
	and then compute the dominant eigenvectors $v_{1}, v_{2}, \dots  v_{\nc}$
	of the square matrix $\cmat = \dmat^T \dmat/\nobs $.
	We overall have a classical computational complexity that is linear in the number of observables
	$\mathcal{O}(\ttot^2 \nobs)$.
	The spectral density function $\lambda(\omega)$ at freqeuncy $\omega$
	is then obtained as the dominant singular value of the spectral cross-correlation
	matrix $\xmat(\omega)$ which we calculate from the dominant $c$ eigenvectors $v_{1}, v_{2}, \dots  v_{\nc}$
	of $\cmat$.
	This step has a complexity $\mathcal{O}(\nc^3 \ttot)$.
\end{statement}

\subsection{Robustness to gate noise}
It is crucial to consider how the ideal signal, as the time-dependent expected value of a Pauli string $S(t_n)  := \langle \psi(t_n) |  P_k |\psi(t_n) \rangle$,
is altered when the time-evolved quantum states $|\psi(t_n) \rangle$ are prepared using a noisy quantum circuit.
We present a simple but powerful proof in \cref{th:error} which states that, under a broad class of typical gate-error
models, the noisy spectrum $\mathcal{F}[S^{noisy}_k]$
decomposes into a sum of the ideal spectrum  $\mathcal{F}[\eta S_k]$
and an additive $\mathcal{F}[(1-\eta) W_k]$ artefact. This
artefact has a magnitude bounded by how close the noise model is to global depolarising (white) noise
and $\eta$ is either a constant or a time-dependent exponential envelope depending on the simulation algorithm used.
\begin{statement}\label{stat:robustness}
	The following are consequences of \cref{th:error}.		
	\begin{itemize}[itemsep=0.em,topsep=0pt,leftmargin=10pt]
		\item \textbf{spectral peak centres are unchanged by incoherent noise}
		\item as long as the circuit error rate is reasonable as $\xi \approx 1$ 
		noise slightly shrinks the intensity of the ideal peaks but does not cause
		a broadening of their lineshape
		\item additive artefacts might appear in the spectrum due to the Fourier transform
		$\mathcal{F}[(1-\eta) W_k]$, however, their intensity is generally bounded and the bounds
		decrease as we increase the system size (number of qubits)
		\item since the additive artefacts $\mathcal{F}[(1-\eta) W_k]$ are specific to each Pauli string $P_k$ and
		assuming they are uncorrelated, i.e., their peaks appear at different frequencies in the spectrum,
		their effect is suppressed by a factor $\nobs$ in the shadow spectrum as we increase
		the number of observables.		
	\end{itemize}
\end{statement}
Refer to \cref{app:noise_rob} for further details. The above properties establish that shadow spectroscopy is generally expected to be highly robust
to gate errors.
Most importantly, a broad class of noise models are guaranteed to not change the peak centres in the shadows
spectrum and thus allow one to estimate the frequency of a peak that would be produced
by a fully error-free simulation circuit -- as long as the circuit error rate is reasonable
$\xi \approx 1$.
We also argue in the Appendix that the approach is equally robust against errors that
go beyond the assumed noise model -- the most concrete evidence we give
is our demonstration on current noisy quantum hardware in \cref{sec:spin},
which clearly confirms robustness.
We additionally note the approach can likely be further improved by randomised
compiling techniques~\cite{gu2022noise} or by generalised Pauli twirling
techniques~\cite{PhysRevA.78.012347, PhysRevA.85.042311, cai2019constructing, cai2020mitigating} that
convert local noise into probabilistic noise~\cite{kim2023evidence}. Furthermore, most error mitigation techniques can also
naturally be combined with classical shadows~\cite{jnane2023quantum,kim2023evidence}.

\begin{figure}
    \centering
    \includegraphics[width=0.48\textwidth]{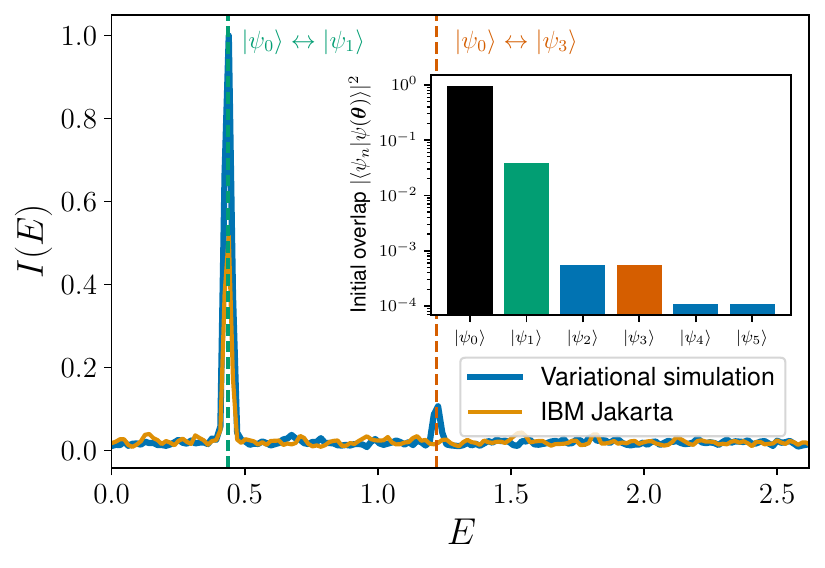}
    \caption{Experimental shadow spectrum of a $6$-qubit disordered Heisenberg chain
    	obtained using the IBM Jakarta QPU.
    	Compared to noise-free simulation of the variational time evolution (blue),
    	the noisy experimental spectrum (orange) correctly resolves
    	exact energy gaps (dashed horizontal lines) between the ground state and first excited state -- albeit with 
    	a lower peak intensity. The error in the gap estimate in the experiment is identical to that in ideal simulation
    	as $2.8\times10^{-3}$.
        The peak corresponding to $E_3-E_0$ has small support $<10^{-3}$ and has 
        low intensity due to the approximate time-evolution algorithm used.
	  	Inset: overlaps $|\langle \psi_n | \psi(\bm{\theta})  \rangle|^2$ between the initial
	  	variational state $\ket{\psi(\bm{\theta})}$ and eigenstates $\ket{\psi_n}$.
    }
    \label{fig:variational_spectrum}
\end{figure}
\section{Applications \label{sec:application}}
Spectroscopy is one of the most employed experimental techniques for probing atomic and molecular physics, commonly used, e.g., to study the geometrical properties or the photochemical/electronic excited-state properties of matter.
If a system's Hamiltonian can be efficiently represented on a quantum device, we can use shadow spectroscopy to discover its spectrum.
Quantum technologies offer different platforms for simulating the time evolution of a quantum system
and in the following we explore several applications targeting key evolutionary stages of quantum computing:
We numerically simulate shadow spectroscopy on NISQ and early-fault tolerant machines as well as on
fully fault-tolerant quantum computers.

\subsection{Shadow spectroscopy for early quantum advantage}
\subsubsection{NISQ -- Spin problems and local Hamiltonians\label{sec:spin}}
In the NISQ era, it is crucial to identify problems that are well-suited to the limitations of NISQ hardware
(such as qubit connectivity and relatively high gate error rates) yet still sufficiently complicated to yield
genuine quantum advantage. The 1D Heisenberg chain with random local magnetic fields and nearest-neighbour interactions as
\begin{equation}\label{eq:spin_ring}
 \mathcal{H}=-J\sum_{j=1}^{N-1}\vec{\sigma}_j\cdot\vec{\sigma}_{j+1}+\sum_j^N h_j\sigma_j^z ,
\end{equation}
has been identified as a promising candidate for early quantum advantage~\cite{luitz2015many,childs2018toward}.
Here $\vec{\sigma}_j=[\sigma_j^x,\sigma_j^y,\sigma_j^z]$ is a vector of Pauli $x$, $y$ and $z$ matrices on the $j$th qubit, $J$ is a coupling constant,
and $h_j\in[-h,h]$ is sampled from a uniform distribution (for disorder strength $h$).
This Hamiltonian is relevant in studies of many-body localization and self-thermalization \cite{luitz2015many,nandkishore2015many}, 
but its phase transitions remain poorly understood due to the difficulty of solving the model on a classical computer.
Therefore, efficiently probing the spectral gap of this model with shadow spectroscopy on NISQ devices may help enhance understanding of key many-body quantum phenomena.

Motivated by circuit depth and gate connectivity limitations of NISQ devices, we apply ansatz-based variational quantum simulation \cite{ying_li,PhysRevLett.125.180501} for the dynamical simulation step of shadow spectroscopy. The state is approximated by an ansatz $\ket{\psi(\bm{\theta})}=U(\bm{\theta})\ket{\bm{0}}$, where $\bm{\theta}$ is a vector of real classical parameters. Furthermore, we use a hardware-efficient ansatz circuit as $U(\bm{\theta})$ (see illustration in \cref{fig:hardware_efficient_ansatz}).
As we detail in \cref{app:variationalsimulation}, at each time step we estimate a small matrix $A_{ij}$
and small vector $C_i$ using the (simulated) quantum device. We then solve the corresponding linear system of equations
using a classical computer to obtain an updated set of parameters $\bm{\theta}$ that best approximate the time-evolved state.

We first apply variational imaginary-time evolution \cite{mcardle2019variational,PhysRevA.106.062416,PRXQuantum.2.030324} to a set of random ansatz parameters 
and  intentionally terminate the algorithm before it could fully reach the ground state to obtain a state
that overlaps with the low-lying excited states. The inset in \cref{fig:variational_spectrum} shows the distribution of eigenstates.
The advantage of near-converged initial states is that a majority overlap with the ground state guarantees resolvable signals that
correspond to energy differences between the ground state and between the excited states (but not among the excited states).
In principle, one could even assign peaks to specific low-lying excitations by repeating shadow spectroscopy on initial states
of varying levels of convergence and comparing peak heights.

In \cref{fig:variational_spectrum} (blue line)
we numerically simulate variational time evolution circuits in
$1000$ time steps and perform shadow spectroscopy
by estimating up to 3-local Pauli operators
(3 batches of $\nshot=50$ snapshots used in each step).
Indeed, the peaks corresponding to low-lying excitations are accurately resolved in a highly NISQ-friendly manner
using shallow quantum circuits.
Also note that expected values of only 3-local Pauli strings were used:
Prior results suggest that local operators indeed yield dominant signal intensities, e.g.,
in translation invariant, gapped Hamiltonians of spin systems
the intensity $|\langle \psi_k | O_i | \psi_l \rangle|$ is an exponentially
decreasing function of the locality of $O_i$~\cite{PhysRevLett.111.080401}.

Next we provide evidence for the ability of our approach to extract accurate gap information from extant
NISQ quantum hardware.
Taking our classically pre-computed time-evolution circuits, we
measure classical shadows of $1000$ time-evolved states on the $7$-qubit IBM Jakarta quantum computer
using the same budget of snapshots as above.
The experimental shadow spectrum in \cref{fig:variational_spectrum} (orange line) clearly resolves the main peak
and demonstrably achieves
the same precision in predicting the energy gap as the pure state simulator without using any error
mitigation techniques -- this is thanks our quite general theoretical guarantees that
only the peak intensity is shrunk by noise which we analyse in the next section.
The second, faint peak in simulation is not resolved in the experimental spectrum; but, as we detail in \cref{app:variationalsimulation},
this peak is actually already suppressed by algorithmic errors from the approximate, variational time evolution,
and here it is likely suppressed to the level of shot noise.
We additionally report QASM simulations using the noise model of the IBM Jakarta chip in \cref{app: hardware}
finding the hardware experiment was actually closer to the ideal noiseless simulation than noisy simulation
-- this could be a result of out-of-date  and inaccurate noise models for the quantum chip.

Moving forward, we expect that shadow spectroscopy can be a major enabler for early quantum advantage for the following reasons.
First, variational circuits are substantially shallower and scale much better than
e.g. Trotter circuits; because the quantum state evolution is stored entirely in the classical
parameters $\bm{\theta}$, the circuit depth is independent of $T$ as opposed to linear in $T$ for Trotter, allowing us in principle to increase the spectral resolution without increasing circuit depth.
Secondly, although we note that optimisation of the dynamical simulation parameters will likely still
require a substantial measurement overhead, we require a significantly lower number of circuit repetitions than what is generally required for state-of-the-art VQE~\cite{endoHybridQuantumClassicalAlgorithms2021, cerezoVariationalQuantumAlgorithms2021a, bharti2021noisy}
where figures beyond $10^6$ shots are quite typical.
Thirdly, several behaviours of variational quantum eigensolvers that are typically considered to be deficiencies are actually advantageous for shadow spectroscopy.
For example, one does not require an ansatz circuit that is sufficiently deep to precisely express the ground state. Instead, we need only be able to express a
superposition of the low lying eigenstates and approximately simulate time evolution using
well-established variational techniques. 
Also, it is well documented that VQE methods typically get `stuck' in one of the exponentially  many local minima around the ground state~\cite{PhysRevX.12.041022,anschuetz2022beyond,nemkov2024barren}. In contrast, an initial state that is only partially converged can be ideal for shadow spectroscopy, since substantial gradients for optimizing dynamics can be guaranteed \cite{drudis2024variational} even when VQE is trapped in a false minimum (though it remains an active question when this corresponds to a good solution to the dynamics).
Finally, although variational simulation unavoidably does not follow the exact time evolution and can thus change the eigenstate composition of the variational state, the appearance of additional spectral peaks nonetheless correspond to physical, albeit unintended, excitations of the system -- which we observe in numerical simulations.
We repeat simulations of spin models for up to 100 qubits in \cref{app:scaling} using tensor-network techniques
and confirm that an increasing system size improves the performance of the present approach.

\begin{figure}
    \centering
    \includegraphics{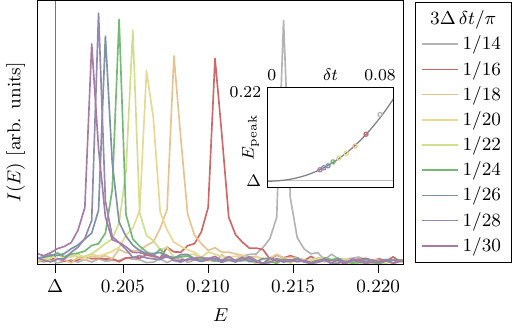}
    \caption{Spectra obtained using Lie-Trotter-Suzuki time evolution for a Fermi-Hubbard model with exact spectral gap $\Delta$.
    Spectra obtained for increasing time step sizes $\delta t$---chosen as whole-numbered fractions of $\pi / 3\Delta$---show
    peaks quite away from $\Delta$ due to significant alogithimc errors.
    Classical shadows of only $50$ snapshots were used to estimate expectation values of all up to 3-local Pauli observables.
    Inset: positions of the peaks as a function of $\delta t$ used;
    fitting a cubic polynomial (to all but the largest step size) and extrapolating to $\delta t \rightarrow 0$ allows us to estimate the exact
    gap $\Delta$ to high precision.
    }
    \label{fig:hubbard_trotter}
\end{figure}

\subsubsection{Early Fault-Tolerance -- Fermi-Hubbard Model}
\label{sec:hubbardmodel}
We now illustrate the utility of classical shadow spectroscopy in the next key evolutionary stage
of quantum computers as early fault-tolerant devices.
These machines will be able to implement significantly
deeper circuits than in the NISQ-era enabling simple Hamiltonian
simulation algorithms such as Trotterisation.
However, due to the
excessive resources required for, e.g., implementing T gates,
gate operations will still incur a certain level of noise
limiting circuit depth thus limiting evolution times and the precision of the simulation (algorithmic errors).

We analyse in detail the resilience of shadow spectroscopy against both algorithmic and
gate errors in another model of great potential for early quantum advantage. We consider the Fermi-Hubbard model as
\begin{equation*}
    H = -t \sum_{\langle i, j\rangle, \sigma}\left(c_{i,\sigma}^\dagger c_{j,\sigma}+ c_{j,\sigma}^\dagger c_{i,\sigma}\right) + U \sum_i c_{i, \uparrow}^\dagger c_{i, \uparrow}c_{i, \downarrow}^\dagger c_{i, \downarrow}
\end{equation*}
where $i$ and $j$ number the lattice sites, $\langle\,\cdot\,,\cdot\,\rangle$ means pairs of neighbouring sites, $\sigma \in \{\uparrow, \downarrow\}$ is
the spin of the fermion, and $c_{i, \sigma}^{(\dagger)}$ are fermionic annihilation (creation) operators at site $i$ with spin $\sigma$.
We obtain the Hamiltonian of qubit Pauli operators $H = \sum_{\ell = 1}^L H_\ell$ through the Jordan-Wigner (JW)transformation
assuming a chain without periodic boundary conditions with parameters $t = 1$ and $U = 2$
~\footnote{Hamiltonian generation and JW mapping were automatically performed using \texttt{openfermion}\cite{openfermion}.}.
We consider a small example system of $12$ qubits ($3 \times 2$ sites) and 
initialise our time evolution in an exact, equal superposition of the ground- and first excited states as $\ket{\psi(0)} \propto \ket{\psi_0} + \ket{\psi_1}$.
This guarantees a single peak in the shadow spectrum at the spectral gap
$\Delta \coloneqq E_0 - E_1$ enabling us to clearly analyse the
robustness of shadow spectroscopy against increasing levels of errors.

\subsubsection{Effects of algorithmic errors}

We first examine how algorithmic errors alone influence the shadow spectrum.
We assume that due to the limited resources (number of qubits and non-negligable gate errors) we need
to resort to a first-order Trotterised time evolution algorithm with limited circuit depth.
Thus, in order to implement a sufficiently long simulation $T=1000 \pi / \Delta$,
we choose a relatively large $\delta t$ which, however, introduces a significant algorithmic error.
\Cref{fig:hubbard_trotter} shows the reconstructed spectra for different values of $\delta t$
and confirms that all peaks are quite far from the exact energy gap due to the algorithmic error.

In \Cref{fig:hubbard_trotter}(inset) we plot the positions of the very accurately resolved peaks
as a function of the step size $\delta t$. By fitting a cubic polynomial 
we extrapolate to $\delta t\rightarrow 0$ obtaining  an excellent approximation of the spectral gap as
$\Delta_\mathrm{est} = 0.2009 \pm 0.0003$, within error margin of the real gap $\Delta = 0.2010$.
Rigorous theoretical guarantees on the efficacy of extrapolating trotter errors are available in Ref.~\cite{rendon2022improved}.
The present extrapolation approach is strictly more powerful
than prior algorithmic error mitigation techniques, including ones for extracting expected values of observables~\cite{PhysRevA.99.012334},
as it allows to directly access an error-mitigated  estimate of the spectral gap
without using, e.g., phase estimation algorithms~\cite{rendon2022improved}.

Furthermore, we expect that shadow spectroscopy is amiable to extrapolation-based error mitigation
techniques for a range of other time-evolution techniques whereby one has control
of certain hyperparameters (such as $\delta t$ in the present case) that influence the
quality of the simulation, such as in qDRIFT~\cite{Campbell_2019} in stochastic evolution~\cite{Ouyang2020compilation}
and beyond.  While in the present example we used a simple initial state to facilitate the analyisis, 
care must be taken when multiple lines are present in the spectrum to accurately monitor changes between 
runs with different time step sizes.

\subsubsection{Effects of gate noise \label{sec:gate_noise}}
In addition to algorithmic errors, we also investigate the influence of hardware noise.
In particular, in an early fault-tolerant device we assume the dominant source of error is
due to the imperfect implementation of T gates and describe an error model in \Cref{app:hubbard}.
In \cref{fig:hubbard_noisy} we confirm that shadow spectroscopy is indeed immune to reasonable levels of shot noise as the relevant peaks at
 $E\approx 0.2$ are very pronounced and centred exactly where the noise-free peaks (dashed horizontal lines)

In particular, \Cref{fig:hubbard_noisy} shows three distinct cases around the error rate at which the result becomes unusable. At a circuit error rate of $\xi\coloneqq \lambda N_\mathrm{gates} \approx 1.0$ and $1.5$, the spectrum shows only one peak, the same as for the noise-free simulation. At $\xi \approx 2$ the peak drastically diminishes in height, and a new component close to $E \approx 0$ emerges. At $\xi=2.5$ finally, the signal of the system at $E \approx 0.2$ disappears, and all the weight apart from the random noise has shifted over to $E \approx 0$.
Indeed, this is  an expected limitation and applies collectively to error mitigation techniques as well: due the exponentially decreasing fidelity $F \approx e^{-\xi}$
the buildup of error escalates rapidly with increasing $\xi$ and makes it impractical to extract useful information~\cite{error_mitigation, koczor2021dominant}. On the other hand, while in NISQ applications one needs to resort to error mitigation techniques to extract useful information,
the present approach is demonstrably robust against gate noise as the peak position is unaffected and only its intensity
is affected by gate noise. 

\begin{figure}
    \centering
    \includegraphics{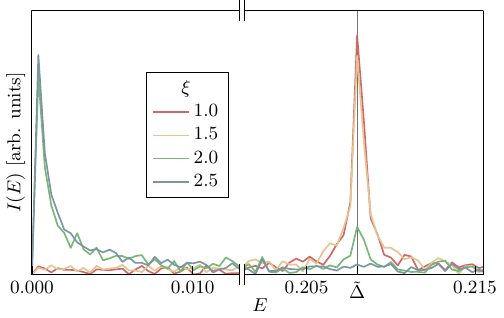}
    \caption{
    Interesting regions (via interruption of the $x$-axis)
    of the spectra obtained for the Hubbard model from simulations including gate noise for an increasing
    number of errors $\xi = \lambda N_\mathrm{gates}$ per a full circuit -- classical shadows of $50$ snapshots were used at every time step.
    The position of the peak without noise using the same Trotter step size as for the noisy simulations is marked as $\tilde{\Delta}$.
    The approach is very robust to errors as guaranteed via \cref{stat:robustness}:
    the peak position is unaffected as we increase $\xi$ but the peak
    disappears as the buildup of errors escalates for large $\xi$.
    }
    \label{fig:hubbard_noisy}
\end{figure}

\subsection{Shadow spectroscopy for mature quantum advantage \label{sec:qchem}}
\subsubsection{Fully Fault-Tolerant -- Quantum Chemistry}
Finally, we explore the fully fault-tolerant quantum computing regime, where algorithmic errors can be suppressed exponentially via sophisticated post-Trotter time-evolution techniques: we describe computation of a molecule's low-lying electronic transitions using shadow spectroscopy. Typically the electronic structure Hamiltonian is specified in a second-quantised form as
\begin{equation}
    \mathcal{H} = \sum_{pq}h_{pq}a_p^\dagger a_q +\frac{1}{2}\sum_{pqrs}h_{pqrs}a_p^\dagger a_q^\dagger a_ra_s,
\end{equation}
and mapped onto multi-qubit Pauli operators using the JW encoding.
Mean-field or Hartree-Fock solutions (Slater determinants) can be computed classically efficiently and correspond to computational basis states.

\begin{figure*}[t]
	\begin{centering}
	        \includegraphics[width=1.0\textwidth]{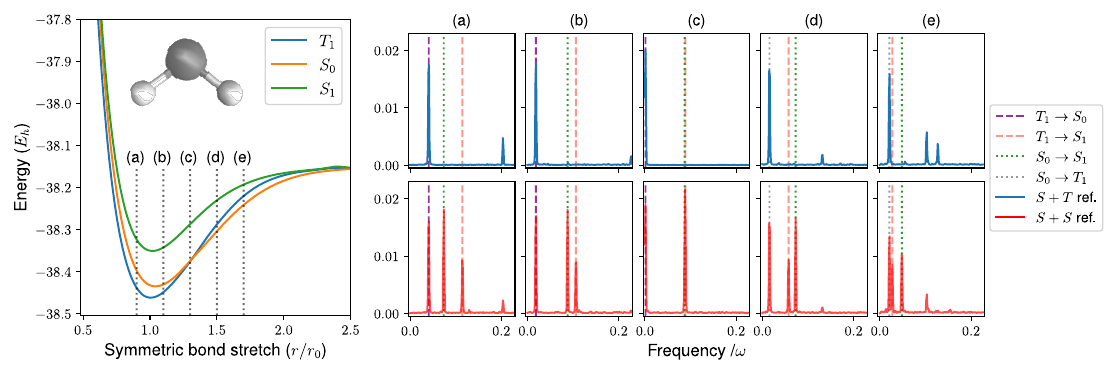}
	\end{centering}
			\caption{\label{fig:Methylenespectra}
			Discovering the electronic spectrum of methylene with (classically emulated) shadow spectroscopy. We fix the bond angle to $101.89^\circ$ and choose 5 symmetric C-H bond geometries between $0.5$ and $3$ times the equilibrium bond length of $r_0=1.1089$~\AA~\cite{SHAVITT1985}. The top row of spectra were generated using an initial state which was a superposition of the lowest energy singlet and triplet determinants~\cref{eq:ST ref}, and the bottom row of spectra a superposition of the lowest and first excited singlet determinants~\cref{eq:SS ref}. The dotted lines are the exact excitation energies from the singlet ground state, purple dashed lines correspond to those from the triplet ground state.
            From classical shadows of $50$ snapshots,
            all up to 3-local Pauli strings were determined with $\nobs=10689$, out of which $2000$
            signals passed a statistical test (they contain a deterministic signal with p values in the range $10^{-20} \leq p \leq 0.06$). Note that a number of high energy peaks are observed too -- assignment of these transitions can easily be achieved in practice using techniques detailed in the main text.}
\end{figure*}

While we discuss details in \cref{app:qchem}, let us briefly summarise why quantum chemistry is
highly amiable to shadow spectroscopy using local Pauli operators. Under the JW-transform, excitation operators map to non-local Pauli strings e.g. the single excitation
$a_p^\dagger a_q \mapsto i A_p^\dagger A_q \prod_{k=p+1}^{q-1} Z_k$, where $A = (X {+} i Y)/2$ is the qubit lowering operator.
In spite of this, we find that in practice the signal intensities of observables
like 2-local Pauli strings e.g. $X_p X_q$, $X_p Y_q$, or indeed any variant where we additionally append one
or more single-qubit $Z$ operators, yield dominant signals.
The reason is that the intensity of these local observables is determined directly by the single-excitation
coefficients in the Configuration Interaction expansion of the time-evolved states, which is the same as if the observable $a_p^\dagger a_q$ was measured instead.

\subsubsection{Methylene spectra discovery}
We now explore a specific quantum chemistry example: The bent methylene molecule CH$_2$ in the minimal STO-3G basis (14 qubits). Because of the challenges in isolating this molecular biradical experimentally, methylene was one of the first instances where classical computational quantum chemistry calculations predicted the properties of a molecule before it was then verified by photochemistry and spectroscopy experimentally~\cite{BoysFoster1960_CH2, Schaefer1986_CH2Perspective, Harrison2009, Sherrill1998}.
It exhibits the following atypical features:
First, its ground state (at equilibrium geometry) has an open-shelled triplet electron configuration
instead of the usual closed-shell singlet.
Second, the singlet and triplet states cross at a certain geometry, called a conical intersection in the excited-state chemistry literature,
and the system becomes degenerate.
We probe these features~\cite{Veis2010_CH2_QC, Veis2014_CH2_QC, sugisaki2021spinchem} by estimating shadow spectra at 5 evenly-spaced geometries 
around the conical intersection, see \cref{fig:Methylenespectra} (a)-(e).

Recall that in the initial state we need only have a reasonable support for eigenstates
indexed $k$ and $l$ to observe the energy gap $E_k-E_l$:
In \cref{fig:Methylenespectra} (top) we use an initial state as a superposition
of the two basis states (Slater determinants) that correspond to the lowest-lying singlet and triplet states as
\begin{equation} \label{eq:ST ref}
    \ket{\psi(0)} = \frac{\sqrt{3}}{2}\ket{11111111000000}+\frac{1}{2}\ket{11111110100000},
\end{equation}
which can be prepared straightforwardly~\footnote{using single-qubit Pauli $X$ gates followed by a single application of a subspace-preserving 2-qubit gate, such as the fSIM gate~\cite{PhysRevLett.125.120504} (see e.g. Ref.~\cite{Sugisaki2019MC} for related strategies for preparing similar initial states).}.
Here we time-evolve these states exactly in $\ttot = 500$ evenly-spaced time intervals with $\Delta t=10$,
and at each time increment determine expected values of all 3-local Pauli strings from
$50$ snapshots of classical shadows.

We plot the resulting shadow spectra for each geometry in \cref{fig:Methylenespectra} (top spectra).
As expected, we find a dominant peak in each spectrum and their positions agree with the exact singlet-triplet gap (energy difference between the $S_0$ and $T_1$ eigenstates) to within error $\pm 0.4$~mE$_h$ (0.25 kcal/mol), as predicted from the Heisenberg scaling $\propto \tfrac{2 \pi}{\ttot \Delta t}$ -- as singlet-triplet gaps are typically on the order of few kcal/mol, the low error of shadow spectroscopy is very promising.  
While physical optical spectroscopy only measures a single observable (e.g. electric dipole moment) where `spin-forbidden' transitions can prohibit resolution of the $S_0 \leftrightarrow T_1$  transition, shadow spectroscopy allows us to probe many different observables, some of which indeed lead to intense peaks as long as the initial state has support for them. 
At geometry (c), despite the near degenerate singlet-triplet intersection, the tiny expected frequency is still observed, confirming that indeed shadow spectroscopy is expected to be highly accurate when used in combination with sophisticated post-Trotter Hamiltonian simulation,
even in the strong correlation regime where multiple eigenstates are often near degenerate.

We demonstrate another valuable feature as the flexibility to the initial state used:
We initialise the time evolution in another superposition of basis states that approximate the two lowest-energy singlet states $S_0$ and $S_1$ as
\begin{equation} \label{eq:SS ref}
	\ket{\psi(0)} = \frac{1}{\sqrt{2}}\ket{11111111000000}+ \frac{1}{\sqrt{2}}\ket{11111110010000},
\end{equation}
In \cref{fig:Methylenespectra} (bottom spectra),
we observe the same peak corresponding to the $T_1 \leftrightarrow S_0$ transition 
suggesting the Slater determinants in \cref{eq:SS ref} have sufficient fidelity with respect to the exact $T_1$ and $S_0$ eigenstates.
We also observe new higher energy peaks: their positions match the exact transition energies of $S_0 \leftrightarrow S_1$ and $T_1 \leftrightarrow S_1$ states within error $\pm0.3$~mE$_h$ (0.18~kcal/mol).

We conclude by noting that the low-lying excited state spectrum of a more complex molecule will be more challenging to decipher.
While in such complex scenarios phase estimation protocols are equally challenging (e.g. due to small overlaps with many eigenstates), the present approach may be helpful in reducing complexity, as different combinations of input states and Pauli strings measured only give rise to specific peaks.
A reasonable strategy might be, as we have demonstrated in \cref{fig:Methylenespectra}, to trial a few initial Slater determinants, linear combinations of them, different choice of Pauli operator observables, and multiple geometries. The difference between these scenarios should let the user assign peaks to energy transitions by elimination and obtain an excited state landscape.

Another strategy is to make use of additional system knowledge:
one could envisage a hybrid quantum-classical learning cycle that uses cheap e.g. density functional or perturbation theory calculations to provide initial state guesses, and from the shadow spectroscopy results update the initial reference or geometry until a confident peak assignment is provided.
One might also consider suppressing high frequency satellite peaks by reducing support for high energy states
such as in the case of imaginary time evolution or variational quantum eigensolvers as we demonstrated in \cref{sec:spin}.
If one instead desires e.g. ionisation energies or electron affinities, the present approach can also in principle be straightforwardly applied by initialising in a superposition of Slater determinants with different electron-numbers (different Hamming-weights).
Further explorations in any of these directions may prove fruitful and offer a new pathway to quantum computing energy gaps in molecules without ever preparing eigenstates in a quantum computer and evaluating static electronic energies.

\section{Further extensions and applications	\label{sec:discussion}}

Our post-processing approach is the combination of relatively simple signal-processing techniques
but they demonstrably work well in practice and guarantee a linear classical computational complexity.
We note that we could use more advanced signal processing techniques and the literature on this topic is rich~\cite{ramirez2008generalization,malekpour2018measures,santamaria2007estimation}.
We could also use compressed sensing to exploit sparsity in the shadow spectrum thereby
reducing the overall number of time-point samplings (the number of timesteps)~\cite{davenport2012introduction,Sherbert_2022}.
Furthermore, while we primarily focused on well-established algorithmic
time-evolution simulation, the presented techniques naturally extend to time evolutions performed with more advanced methods that correct the algorithmic errors incurred at the unitary level, as well as on analog quantum simulators
combined with recent analog classical shadow techniques~\cite{analogShadow}

Beyond classical shadows, we considered problem-specific operator pools and constructed them explicitly
for the case of quantum chemistry. This lends itself to a number of straightforward extensions. 
First, we could similarly construct problem specific operator pools of non-local Pauli strings
and use derandomisation protocols~\cite{shadows_derandomization} or shallow Clifford
circuits~\cite{akhtarScalableFlexibleClassical2022, bertoniShallowShadowsExpectation2022}
to estimate expected values efficiently.
Second, the operator pools we constructed for fermionic problems can be grouped into large mutually
commuting groups and we could measure these using well-established techniques~\cite{Crawford2019,yen2020measuring,jena2019pauli,gokhale2020n}.
Third, we could also construct an operator pool of fermionic operators and use
fermionic shadows~\cite{wanMatchgateShadowsFermionic2022, zhaoFermionicPartialTomography2021} to estimate expected values.

A number of further applications and generalisations are apparent.
First, we can apply the present approach to an initial state that is relatively close to the ground state, e.g, a mean-field approximation
for fermionic systems. By inspecting the shadow spectrum, one can semi-quantitatively
determine how many and which excited states contribute with dominant weights. 

Second, we could also use the present approach to diagnose gaps in adiabatic quantum computing:
As in adiabatic evolution the quantum state closely follows the ground state,
the shadow spectrum has peaks only at excitation energies between the ground and excited states.
One can thus uniquely identify the spectral gap and determine whether it is sufficiently different
from zero -- effectively verifying the validity of the adiabatic theorem.

Third, while in the present work we primarily focused on determining energy gaps, it
is straightforward to extend our approach to the direct estimation of energy levels:
One usually proceeds by embedding the Hamiltonian into a larger system by adding
an additional ancilla qubit. For example, in phase estimation one applies a
controlled evolution while in NISQ-friendly phase estimation the Hamiltonian 
is appended with additional Pauli Z operators~\cite{clinton2021phase} --
finding operators that anticommute with the Hamiltonian are also equally effective~\cite{gnatenko2022energy}.
These approaches, however, come with increased quantum resources and may require bespoke architectures.

\section{Discussion and Conclusion\label{sec:conclusion}}

In this work we developed a general technique that can be used to extract energy differences
between discrete energy levels of a quantum system: information that is of key importance 
for many applications in chemistry, materials science and beyond. Our approach can be viewed as quantum-digitisation of spectroscopy:
estimating observable properties of a time-evolved
quantum system from its classical shadows allows us to learn fundamental frequencies in the
time evolution as intense peaks in a shadow spectrum -- these peaks precisely
correspond to energy differences in the system's Hamiltonian.
The primary limitation of the approach is that a sufficiently good initial
state must be prepared that approximates the relevant eigenstates -- this is unavoidably hard in general
but is a common requirement for most classical and quantum computational methods which we detail in \cref{app:limitations}.

Classical shadows allow us to 
estimate a large number ($\nobs \approx 10^8$ is realistic) of time-dependent observables
which is beneficial for at least three reasons.
First, we can efficiently search for observables of high intensity, i.e., 
we estimate all $q$-local Pauli strings and select only the observables
that give rise to intense signals using statistical correlation tests.
Second, the effect of gate noise is potentially suppressed as $\propto \nobs^{-1}$.
Third, the signal-to-noise ratio (shot noise) in the spectrum is increased as $\propto \sqrt{\nobs}$.
While the classical computational complexity of our approach is linear and the required 
quantum resources are only logarithmic in $\nobs$, we verified in tensor-network simulations 
for up to 100 qubits that indeed our approach is scalable and its performance improves for an
increasing number of qubits.
We also prove that the approach is immune to reasonable levels of gate errors as the ideal spectral peak positions
are not altered by incoherent noise -- which we verified in real quantum hardware demonstrations.

Furthermore, we provide a broad range of practically motivated demonstrative examples that illustrate how our
approach indeed
(a) requires no additional quantum resources beyond the ability to simulate time-evolution
and to extract classical snapshots (randomised measurements);
(b) is highly robust to gate noise and shot noise,
(c) admits a Heisenberg limited precision;
(d) is amiable to a simple yet effective algorithmic error
mitigation technique.
Shadow spectroscopy is also very versatile as it can be used for a large variety of
quantum systems and Hamiltonians, and in principle with any quantum simulation platform.
These include analog quantum 
simulators, NISQ machines and early fault-tolerant quantum computers.
	For example, ion traps \cite{kokail2021entanglement} can simulate time
	evolution of non-trivial spin models in an analog, non-gate based fashion and can perform single-site
	Pauli measurements, making them an ideal platform for shadow spectroscopy.

The amount of absolute quantum resources required for the present approach can be very low;
in both noisy hardware experiments and numerical simulations a few dozen circuit repetitions per time step were sufficient to measure shadow spectra.
Thus we expect the present approach requires a moderately increased number of circuit repetitions
compared to even mature phase-estimation protocols which is, however, a worthwhile tradeoff
due to the above noted advantages --  making the present approach a highly competitive alternative. 
Furthermore, shadow spectroscopy uses many orders of
magnitude fewer shots than typical NISQ applications where figures well beyond $10^6$ shots for a single eigenvalue estimation are very common
-- of course simulating time evolution with variational techniques may still require high repetition as in \cref{sec:spin}.

Shadow spectroscopy has numerous natural applications ranging from high-energy physics to materials science
to modelling molecular systems.
In fact, we identify quantum chemistry as a sweet spot for our approach:
popular VQE algorithms typically suffer from grave performance deficiencies for quantum chemistry Hamiltonians
such as astronomical sampling (circuit repetition) costs, barren plateaus and relatively deep circuit depths required -- extracting the spectrum additionally requires computing excited states which exacerbates these drawbacks.
In contrast, we prove analytically and demonstrate experimentally that
local Pauli strings give rise to intense signals from which we can estimate energy differences
as long as one can prepare good initial states and time evolve them;
furthermore, the efficient  fermionic shadows-measurement routine is very suitable for shadow spectroscopy
for molecular systems~\cite{wanMatchgateShadowsFermionic2022, zhaoFermionicPartialTomography2021}.

Shadow spectroscopy may also be useful in entanglement spectroscopy~\cite{kokail2021entanglement,zache2022},
where the transition energies in a learned ``entanglement Hamiltonian'' correspond to the entanglement
spectrum of a subsystem. Our method may provide a way of observing that spectrum experimentally with many fewer measurements.

We believe the present approach can be a major enabler for practical quantum advantage both on NISQ devices
as well as on early fault-tolerant quantum computers due to its versatility, efficiency and practicality.
Furthermore, the present approach motivates deeper questions:
can one find more direct ways to analyse time-dependence from classical shadows circumventing
the intermediate step of estimating expected values of Pauli operators? In the present work we resorted to local Pauli strings and found that in all examples up to 3-local Pauli strings gave rise to sufficiently informative signals about low-lying energy transitions in practice (which we verified up to 100 qubits).

\section*{Acknowledgments}
The authors thank Simon Benjamin for his support and help throughout all stages of this work.
The authors thank Tyson Jones for developing efficient simulation code for classical shadows in QuESTlink.
The authors are grateful to Zsuzsanna Koczor-Benda for helpful comments on quantum chemistry methods.
The authors thank Ryan Babbush, Thomas O'Brien, Jarrod McClean, Daniel Marti-Dafcik, Sam McArdle and Norbert Schuch for helpful comments.
Render of the methylene and acetaldehyde molecules were supplied by Annina Lieberherr.
The authors acknowledge the EPSRC Hub grant under the agreement number
EP/T001062/1 for hardware provision.
The authors also acknowledge funding from the
EPSRC projects Robust and Reliable Quantum Computing (RoaRQ, EP/W032635/1)
and Software Enabling Early Quantum Advantage (SEEQA, EP/Y004655/1).
H.H.S.C. acknowledges the support of the Croucher Foundation, Hong Kong.
M.L.G. acknowledges the support of a Rhodes Scholarship.
B.K. thanks the University of Oxford for
a Glasstone Research Fellowship and Lady Margaret Hall, Oxford for a Research Fellowship.
The numerical modelling involved in this study made
use of the Quantum Exact Simulation Toolkit (QuEST), and the recent development
QuESTlink\,\cite{QuESTlink} which permits the user to use Mathematica as the
integrated front end, and pyQuEST\,\cite{pyquest} which allows access to QuEST from Python. We are grateful to those who have contributed
to all of these valuable tools. 
The authors would like to acknowledge the use of the University of Oxford Advanced Research Computing (ARC) facility\,\cite{oxfordARC} and the IBM Quantum services in carrying out this work.
B.K. conceived the idea, produced proofs and performed scaling analyses, H.H.S.C. performed the methylene simulations, spin-chain hardware experiment and developed quantum-chemistry specific ideas, R.M. performed ideal and noisy simulations of the Hubbard model, M.L.G. performed ideal and noisy variational simulations of the spin-chain. All authors contributed to writing the manuscript.

\appendix

\section{Summarising limitations of shadow spectroscopy}
\label{app:limitations}

The main limiting factor of shadow spectroscopy is that it requires an initial state
with a sufficient overlap with the eigenstates whose energy differences one aims to extract, i.e.,
typically the lowest-lying eigenstates in practice. Obtaining good initial states is, however,
a challenging task which can be a bottleneck for even advanced, fault-tolerant
phase-estimation protocols~\cite{garnetchan2022}. This challenge can, indeed, be attributed to general,
exponential hardness results for finding ground or eigenstates~\cite{bookatz2012qma}.

As such, applying quantum or classical heuristics becomes crucial for obtaining approximations of
low-lying eigenstates and their state-preparation circuits, such as by using tensor-network methods~\cite{jamet2023anderson}.
Regarding quantum heuristics, in the present  work we explored
initialisation through VQE in which case hardness results manifest in the difficulty of circuit training:
for deep, randomly initialised variational circuits one indeed suffers from exponentially increasing
training costs due to barren plateaus~\cite{mcclean2018barren} or due to exponentially many local traps~\cite{anschuetz2022beyond,nemkov2024barren}
and thus the efficacy of obtaining good initialisation for shadow spectroscopy boils down
to finding well-motivated, problem-specific ans\"atze for which efficient training is possible.

The other main limiting factor of algorithmic shadow spectroscopy is the accuracy of Hamiltonian simulation:
as seen from the Heisenberg chain and Fermi-Hubbard model examples, unitary-evolution simulations are
only approximate, introducing algorithmic errors that shift the positions of the peaks.
However, in this work we show that these errors can be mitigated (see \cref{sec:hubbardmodel})
and indeed theoretical guarantees are available in the literature~\cite{rendon2022improved}. As mentioned, advanced time evolution approximations that have lower algorithmic errors are also constantly being invented.

Furthermore, gate noise and shot noise only affect the signal-to-noise ratio (peak intensities),
not peak position accuracy as long as the noise model is probabilistic, e.g., Pauli noise, as we prove.
However, our proofs do not apply to coherent gate errors or to error models that are not of the form
of~\cref{eq:noise_model}, such as damping errors. We do not expect this to be a major drawback as in
typical experiments one applies twirling operations to convert the aforementioned
non-probabilistic noise sources into a form that complies with~\cref{eq:noise_model}.
Indeed, as we demonstrate numerically, even without twirling operations the approach is highly
noise robust even for damping errors that violate our noise-model assumptions in~\cref{eq:noise_model}.
In addition, these claims are in good agreement with our hardware experiment.

Our approach achieves Heisenberg scaling in the sense that the evolution time required
for an energy gap $\Delta$ to be resolved is $T\sim 1/\Delta$ (this is different from the $1/\epsilon$ scaling used in quantum metrology);
It follows that quantum resources as both circuit depth and the number of shots
 (assuming at least one shot taken per timestep)
required scales similarly as $\sim 1/\Delta$. Quantum protocols in which the required quantum resources
scale inversely proportionally with the precision are often refereed to as admitting Heisenberg scaling.
This can be contrasted with the standard shot-noise scaling of energy estimation in VQE where the precision
scales with the inverse square root of the number of shots.

However, a limitation in this context is spectral gap closure, i.e., the smaller the spectral gap the deeper the necessary circuits.
For example, in systems with `strong correlation', such as the Fermi-Hubbard model where $U/t \ll 1$
or at chemical bond dissociation limits in molecules,
low-lying eigenstates may get exponentially close to each other,
forming a dense band with degenerate or near-degenerate states. 
Note, however, that resolving individual eigenvalues becomes equally difficult with
phase-estimation protocols in these cases.
Other adversarial examples which may give exponentially vanishing signal
intensities include random quantum states due to the exponential concentration of
eigenvalues of traceless Hamiltonians around $E=0$.
As such, the primary application of the present approach is to accurately
resolve excitation energies in gapped quantum systems where the gap closes 
no faster than polynomially.
Indeed, our numerical results in \cref{sec:qchem} demonstrate that the technique is able to
resolve the energy gap regime near a conical intersection.
While even answering the question whether a broad class of quantum systems is gapped
in the thermodynamic limit is an undecidable problem~\cite{cubitt2015undecidability},
most practically relevant problems involve gapped systems, e.g., molecular systems.

Finally, we note that 
while shadow spectroscopy is formally similar to physical spectroscopy, it is strictly more general.
The reason is that the intensities of the peaks in conventional spectroscopy
are determined by the transition dipole moment, in the present case we have
significant flexibility in choosing various other operators that we can estimate from classical shadows. 
Thus the primary aim of shadow spectroscopy is not necessarily to faithfully simulate conventional spectroscopy, 
but to rather estimate differences of eigenvalues making as frugal use of quantum resources as possible.
Consequently, the approach can resolve peaks corresponding to spin-forbidden transitions, which will
not be observed in optical spectroscopy.
In contrast, if the aim was to faithfully simulate conventional spectroscopy, shadow spectroscopy 
can still in principle be used by estimating specifically the transition dipole moment from shadows.

\section{Comparison to literature\label{app:lit}}

\subsection{Simulated spectroscopy\label{sec:simspec}}
A number of works have used similar ideas to track the time-dependent expected value of some specific observable
to reveal transition energies in certain spin problems -- in direct analogy with spectroscopy experiments.
First, Ref.~\cite{yoshimura2014diabatic} proposed to simulate the evolution of a spin system under a diabatically
ramped magnetic interaction, whereby
estimating the time-dependent signal of an observable allows to identify transition energies.
Second, Ref.~\cite{gnatenko2022energy} proposed to measure the time-dependent signal of a single
observable $O$ that anticommutes with the Hamiltonian $\{\mathcal{H},O\} = 0$ -- for certain simple spin problems
such an anticommuting observable can be constructed straightforwardly.
Third, coupling the system to an ancilla qubit via an interaction term and measuring the time-dependent expected value of the Pauli $Z$
operator can also be used to identify spectroscopic transitions \cite{PhysRevResearch.4.043106}.

Indeed, problem specific observables, as identified by the aforementioned works for spin problems, can be readily adopted
for our approach.
While the present work can indeed similarly be viewed as a variant of spectroscopy, it distinguishes itself as we estimate
the signal of a very large number of observables using only very few circuit repetitions. Furthermore, the present approach
can be applied in principle to any problem Hamiltonian or at least to ones that give non-negligible signal upon
measuring local observables -- while noting that non-local operators may also be measured efficiently by other variants of classical shadows.

\subsection{Quantum algorithms for excited states\label{sec:quantumexcite}}
The low-lying excited state spectrum of a system can of course be directly calculated by evaluating the ground and excited state energy levels. A great plethora of variational methods e.g. deflation~\cite{higgott2019variational, chan2021molecular}, equations-of-motion~\cite{EOM-VQE2, EOM-VQE}, variational imaginary time evolution~\cite{jones2019variational}, subspace-search \cite{nakanishi2019subspace} have been proposed. Most relevant to our method are those which use a quantum computer to generate subspaces, which we discuss presently.

In subspace expansion~\cite{mccleanSubspaceExpansion} one can discover low-energy excited states from an estimated ground state $| \tilde{\psi}_G \rangle$ by applying low-weight operators as excitations to this initial state.
Typically Pauli operators are used to produce a new set of states $\ket{\psi_k}=\mathcal{P}_k \ket{\tilde{\psi}_G}$ for calculating the overlaps $H_{kj}=\bra{\psi_k}\mathcal{H}\ket{\psi_j}$ and $S_{kj}=\braket{\psi_k|\psi_j}$.
Finally, diagonalising $H_{kj}$ then reveals better ground-state energies than that of $| \tilde{\psi}_G \rangle$.
In our approach in \cref{statement:intenisty}, the intensity
of the signals is expressed as $I_{kl} = c_k^* c_l   \langle \psi_k | O | \psi_l \rangle$.
As such, we aim to apply operator pools that give non-negligable overlaps $\langle \psi_k | O | \psi_l \rangle$ and
thus we can draw a loose analogy with subspace expansion where similar operators are required.
However, the difference in efficiency between the present approach and
subspace expansion is crucial: in subspace expansion one estimates the overlaps one-by-one resulting in
a linear complexity $\mathcal{O}(\nobs)$ whereas our effective use
of classical shadows reduces the quantum resources to only logarithmic $\mathcal{O}(\log \nobs)$.
Nevertheless, literature on subspace expansion will
be relevant in identifying problem specific observables $O$ that give significant overlaps which 
we have considered in \cref{sec:qchem}, e.g.,
excitation operators in quantum chemistry or certain local observables in spin problems.
Furthermore, a significant advantage of the present approach is that explicit knowledge of which operators give intense signals is not necessary as we can efficiently
reconstruct all $q$-local Pauli strings and identify the useful signals via statistical autocorrelation tests
that we discussed in \cref{sec:postproc}.

Methods that generate instead Krylov~\cite{Kirby2022Lanczos} and Krylov-like subspaces~\cite{motta2020determining, Klymko2022} are, in comparison, more loosely related to our approach. These algorithms generate a subspace by applying to a supplied initial state either (a) powers of the Hamiltonian or (b) real- or imaginary-time Hamiltonian evolution steps which approximate powers of the Hamiltonian in the small time-step limit as per the Taylor expansion of the matrix exponential. Computing the subspace Hamiltonian and overlap matrix elements on a quantum computer and diagonalising it classically gives approximations to eigenvalues of the modelled Hamiltonian. While our approach features time evolution of an initial state as well as classical signal post-processing and matrix diagonalisation, it does not have the burden of evaluating Krylov matrix elements, nor does it require an initial state with good ground state overlap.

\subsection{Variants of phase estimation}
A number of works aim to estimate the expected value $\langle \psi |  e^{-i t \mathcal{H} }  |\psi \rangle$
or a signal that is effectively equivalent to it.
First, 
in NISQ to early-fault-tolerant iterative phase estimation protocols~\cite{clinton2021phase}, the autocorrelation signal $\langle \psi |  e^{-i t \mathcal{H} }  |\psi \rangle$
is estimated by applying time evolution steps controlled by a single ancilla and then measuring out the ancilla in shots; the probability of observing the $\ket{0}$ outcome is the autocorrelation signal.
This signal contains frequencies precisely at eigenenergies and with amplitudes determined by the
fidelity of the initial state with respect to the eigenstate -- the approach thus requires a reasonably
large initial overlap to the targeted eigenstates due to shot noise.
A group of recent methods called statistical phase estimation builds on this single ancilla iterative phase estimation protocol and use sophisticated classical signal processing techniques to extract as much information as possible from the single ancilla signal, thereby bringing down the sampling cost and circuit depth for generating said autocorrelation signal on a quantum computer -- see e.g. Refs.~\cite{Wang_2023, Lin_2022, Ding2023simultaneous}. 
Bespoke architectures have also been developed in
Ref.~\cite{clinton2021phase} to reduce quantum resources required for a controlled evolution.

Problem-specific techniques have also been developed for quantum chemistry applications. For example, phase estimation for investigating elusive high energy excitations~\cite{Bauman2021_highenergyex} have been proposed. Refs.~\cite{sugisaki2021bayesian,sugisaki2021quantum} use Hadamard-test circuits to estimate expected values of excitation operators -- and ancilla-free variants are also discussed.
Indeed, in \cref{sec:qchem} we discuss such problem specific operators and construct a large number of these that
we can simultaneously measure with powerful classical shadow techniques.

Other approaches seek to generate the autocorrelation signal without the ancillary qubits. Ref.~\cite{russo2021evaluating} proposes to create a superposition of two eigenstates and performs a Ramsey-type experiment
by applying a time evolution. The approach alleviates the use of a controlled evolution but requires a
$\mathrm{Prep}$ and a $\mathrm{Prep}^\dagger$ oracle that prepare a superposition of two eigenstates.
Besides similarly requiring no controlled evolution, another significant
advantage of our approach is that it requires no explicit knowledge of a state-preparation oracle.
In fact, the present approach is very flexible with respect to the supplied initial state---as we demonstrated in numerical
simulations---and can even work with input states that are a superposition of a large number of eigenstates.

Finally, the present approach can also be compared to Ref.~\cite{PhysRevA.93.062306}, whereby one first applies a Haar-random circuit to the $\ket{0}$ state, applies a time evolution operator, applies the inverse of the random circuit and finally measures the survival probability as probability of the $\ket{0}$ state. While we similarly apply randomised measurements, the advantage of our approach is that we need only use single qubit rotations to obtain classical shadows. 

In summary, these prior works estimate a single autocorrelation signal that may require cumbersome quantum resources such as controlled
time-evolution or state-preparation oracles, 
and then perform classical post-processing to extract eigenvales from the autocorrelation signal.
In contrast, here we only
use logarithmic quantum resources to cheaply estimate a large number of operator expectation value signals, which we then classically analyse
in a non-trivial post-processing step to reveal energy gaps.

\section{Noise robustness\label{app:noise_rob}}
\subsection{Noise model assumptions~\label{app:noise_model}}
We recollect properties of a general class of noise models from \cite{koczor2021dominant}:
Most typical noise models used in practice, such as depolarising, bitflip or dephasing noise,
admit the following probabilistic interpretation: 
a noisy quantum gate $\Phi(\rho)$ acts
with probability $1-\epsilon$ as a noise-free operation $U$ 
and with probability $\epsilon$ as an erroneous execution of the  gate as
\begin{equation}\label{eq:noise_model}
	\Phi_k(\rho) = (1-\epsilon) U_k \rho U_k^\dagger + \epsilon \Phi_{err} (U_k \rho U_k^\dagger).
\end{equation}
Here $U_k$ is the $k^{th}$ ideal quantum gate in a quantum circuit and
the CPTP map $\Phi_{err}$ happens with probability $\epsilon$ and
represents all errors that happened during the execution of the gate.

A quantum circuit is then a composition of a series of $\nu$ such quantum gates
and when this circuit is applied to any reference state it prepares the density matrix
\begin{equation} \label{eq:decomposition}
	\rho = \eta \rhoid + (1-\eta) \rhoerr.
\end{equation}
Here $\rhoid:= |\psi \rangle \langle \psi |$ is the ideal noise-free state that the ideal circuit would prepare,
$\rhoerr$ is an error density matrix that contains a mixture of all error events in the circuit
and $\eta = (1-\epsilon)^\nu$ is the probability that none of the gates have undergone errors. 

\subsection{Robustness to gate noise}
\begin{theorem}\label{th:error}
	We assume that a series of noisy quantum circuits are used to prepare the series of time-evolved quantum states $\rho(t_n)$
	under the general noise model described in \cref{app:noise_model}.
	Any measured time-dependent signal, as the expected value of a Pauli string $P_k$, then decomposes as
	\begin{equation*}
		S^{noisy}_k(t_n) = \eta S_k(t_n) + (1-\eta) W_k(t_n).
	\end{equation*}
	The first term $S_k(t_n)$ is the ideal signal while the second term
	is a noise contribution 
	whose amplitude is bounded $|W_k(t_n)| \leq \lVert \rhoerr(t_n) - \mathrm{Id}/d \rVert_1$ 
	via the distance from the maximally mixed state, i.e.,  the error term is zero for global depolarising noise.
	Depending on the simulation algorithm used, the weight $\eta$ is either a constant
	or a time-dependent exponential envelope
	and thus the Fourier transform $\mathcal{F}[\eta S_k](\omega)$ has peaks centred
	exactly at peaks of the ideal spectrum $\mathcal{F}[S_k](\omega)$.
\end{theorem}
\begin{proof}
	We define the noisy time-dependent signal as 
	\begin{equation*}
		S^{noisy}_k(t_n) = \tr[ \rho(t_n) P_k ],
	\end{equation*}
	which is the expected value of a Pauli string $P_k$ as estimated, e.g., via classical shadows.
	We assumed that the time-evolved quantum states $\rho(t_n)$ are prepared by a series
	of quantum circuits that are described by the error model in \cref{app:noise_model}
	such that every quantum state decomposes as
	\begin{equation*}
		\rho(t_n) = \eta(t_n) \rhoid(t_n)  + [1-\eta(t_n)] \rhoerr(t_n).
	\end{equation*}
	Here the noise-free quantum state $\rhoid(t_n) = U(t_n)|\psi_0 \rangle \langle \psi_0 |U^\dagger(t_n)$ is prepared by the noise free
	simulation circuit $U(t_n)$;
	it has a weight given by the probability that none of gates in the circuit has undergone a noise event
	$\eta(t_n) =  (1-\epsilon)^{\nu(t_n)}$; this weight depends on the per-gate error rate $\epsilon$ and
	on the number $\nu(t_n)$ of gates that are used to prepare the time-evolved quantum state $\rho(t_n)$.
	It immediately follows that the noisy signal decomposes as
	\begin{equation*}
		S^{noisy}_k(t_n) = \eta(t_n) S_k(t_n) + [1-\eta(t_n)] W_k(t_n),
	\end{equation*}
	where the noise signal $W_k(t_n) := \tr[ \rhoerr(t_n) P_k ]$ can be bounded generally as
	\begin{align*}
		|\tr[ \rhoerr(t_n) P_k ] | &= |\tr[ \rhoerr(t_n) P_k ]  - \tr[ \tfrac{\mathrm{Id}}{d} P_k ] |\\
		&\leq \lVert \rhoerr(t_n) - \tfrac{\mathrm{Id}}{d} \rVert_1.
	\end{align*}
	Above we have used that Pauli strings are traceless and have unit operator norm.
	We have also used the general bound of ref~\cite{koczor2021dominant}
	which  can be applied 	 to any observable $O$ as
	\begin{equation*}
		| \tr[O \rho_1] - \tr[O \rho_{2}] | 	\leq  \lVert O \rVert_\infty  \lVert \rho_1 - \rho_{2} \rVert_1.		
	\end{equation*}

	Finally, we derive the time-dependent weight $\eta(t_n)$:
	When the time evolution is simulated by a variational circuit, whose circuit structure is
	constant throughout the simulation
	as in \cref{sec:spin}, then the weight is a time-independent constant $\eta =  (1-\epsilon)^\nu$.
	This constant is well approximated as~\cite{koczor2021dominant,dalzell2021random,whitenoise}
	\begin{equation}
		\eta = (1-\epsilon)^\nu =  e^{-\epsilon \nu} + \mathcal{O}(\epsilon^2 /\nu).
	\end{equation}
	As such, the spectrum of the ideal signal $S_k(t_n)$ is only rescaled by 
	a constant factor that depends exponentially on the expected number of
	errors in the full variational circuit as 	$\xi = \nu \epsilon$.
	
	On the other hand, when the time evolution is simulated
	by product formulas as in \cref{sec:hubbardmodel} then the weight is time dependent as $\eta(t_n) =  (1-\epsilon)^{n N_l}$
	where $N_l$ is the number of quantum gates in a single layer of the product formula circuit.
	This is well approximated by a time-dependent exponential envelope function as
	\begin{equation*}
		\eta(t_n) = (1-\epsilon)^{n N_l} =  e^{-\epsilon N_l n } + \mathcal{O}(\frac{ \epsilon^2 } {n N_l} )  .
	\end{equation*}
	The exponential envelope in the ideal part of the signal $\eta(t_n) S_k(t_n)$
	does not change the peak centre in the Fourier transform but only broadens
	it for the following reason:
        We consider a time-continuous function $\eta(t) S_k(t)$ from which
        the signal $\eta(t_n) S_k(t_n)$ are samples evaluated at times $t_n=n \Delta t$.
        The Fourier transform of $\eta(t) S_k(t)$ can then be calculated analytically as
	\begin{equation} \label{eq:lorentzian}
		\mathcal{F}[ e^{- \alpha t } S_k(t)](\omega)  \propto  	\mathcal{F}[S_k](\omega) \ast \frac{\alpha}{\alpha^2 + \omega^2},
	\end{equation}
	where $\alpha := \frac{\epsilon N_l}{\Delta t}$.
	The above is simply a convolution of the ideal spectrum $\mathcal{F}[S_k](\omega)$ with a Lorentzian function
	of width $\alpha$ -- this convolution does
	not change the centre of the peaks but only broadens their lineshape.
\end{proof}

One could compare the above results to error-robustness guarantees of Ref.~\cite{gu2022noise}.
Specifically, Ref.~\cite{gu2022noise} assumed an error model $\mathcal{U} \mathcal{E}$ that factorises into a product of
the ideal channel $\mathcal{U}$ and a noise channel $\mathcal{E}$, and
proved for small errors $\lVert \mathcal{E} \rVert \ll 1$, i.e., for small circuit error rates $\xi \ll 1$,
that in first order perturbation the phase evolution is unchanged
by the noise and only the amplitude of the signal decreases.
\cref{th:error} presents considerably stronger error-robustness guarantees that apply
to finite error rates $\xi$ and to almost all typical error channels used in practice.

\subsection{Approximate bounds on the error term}

The above result establishes that the noisy signal $S^{noisy}_k(t_n)$ is a superposition of two signals:
The first component is the ideal signal $\eta S_k(t_n)$ up to a rescaling factor $\eta$ which 
however, does not change the centre of peaks in the spectrum. The second component is
a noise signal $(1-\eta) W_k(t_n)$ that may contribute artefacts to the spectrum. 
However, the magnitude of this error signal is bounded by the distance from white (global depolarising) noise
$\lVert \rhoerr(t_n) - \tfrac{\mathrm{Id}}{d} \rVert_1$.

While it can be proved rigorously that in random circuits white noise is approached for
an increasing circuit depth, the statement was recently analysed for a broad class of
practical shallow circuits~\cite{dalzell2021random,whitenoise}. In particular, the exact
formula for random circuits was found to be a good
approximation for practical circuits -- albeit with a potentially large $a \approx 0.1$
problem-dependent prefactor. Let us now adapt bounds of \cite{dalzell2021random,whitenoise}
obtained for random circuits of $\nu$ gates as
\begin{equation}\label{eq:appr_bound}
(1-\eta(t_n)) 	\lVert \rhoerr(t_n) - \tfrac{\mathrm{Id}}{d} \rVert_1  	\approx
\frac{ a \times  e^{-\xi} \xi } {\sqrt{\nu} }.
\end{equation}

In the case when variational circuits of fixed depth are used to simulate time evolution,
as in \cref{sec:spin}, the
above bound is constant throughout the simulation.
We thus find that the error component of the signal is globally bounded approximately 
as $|[1-\eta(t_n)] W_k(t_n)| \lessapprox \alpha \nu^{-1/2}$ assuming $\xi \leq 1$.  
Given the $\nu^{-1/2}$ dependence on the number of gates,
we can expect that the artefacts in the spectrum due to noise diminish
as we increase the systems size due to the  increasing number of gates $\nu$.

Let us now consider product formulas, as in \cref{sec:hubbardmodel}, whereby the number of gates is 
time dependent as $\nu(t_n) = N_l n $ and thus yields the time-index-dependent 
upper bound in \cref{eq:appr_bound} as $a \epsilon e^{-\epsilon N_l n} \sqrt{ N_l n} $.

Let us assume a worst-case scenario whereby the error term saturates its above upper bound and thus
\begin{align*}
	S^{noisy}_k(t_n) &= \eta(t_n) S_k(t_n) + [1-\eta(t_n)] W_k(t_n)\\
	&= 
	 e^{-\epsilon N_l n} [ S_k(t_n) +  a  \epsilon \sqrt{N_l n}  ],
\end{align*}
whereby inside the square brackets  we have a sum of the ideal signal and an additive worst-case noise.
The noise component is a signal that grows proportionally with the square-root of the time index $n$.
We can compute the square of the $l_2$ norm of the error component
of the signal as
\begin{equation*}
\lVert \alpha  \epsilon \sqrt{N_l n} \rVert_2^2 = a^2  \epsilon^2 N_l \sum_{n=1}^{\ttot}  n \approx a^2  \epsilon^2 N_l \ttot^2.
\end{equation*}
Given the $l_2$ norm is invariant under the Fourier transform, we find that while the $l_2$ norm
of the ideal signal $\lVert S_k(t_n) \rVert_2$ is proportional to the peak height in the spectrum,
the total spectral contribution of the artefacts is given by the
$l_2$ norm as $a  \epsilon \ttot \sqrt{N_l} = a \xi / \sqrt{N_l}$.
As such, under the above assumptions, the contribution of the error term to the spectrum diminishes as we increase
the system size due to the increasing number of gates in a layer $N_l$ of a
product formula but while assuming a constant $\xi$.

\subsection{Lineshape broadening}
In case when the time evolution is simulated by product formulas then the ideal signal is
multiplied by an exponential envelope 
which leads to a broadening of the lineshapes in the shadow spectrum via \cref{eq:lorentzian}.
However, as long as the total circuit error rate is reasonably small as $\xi \approx 1$ the broadening is not
a limiting factor in practice. In particular, the Lorentzian
width of the peak in \cref{eq:lorentzian} is given by the factor
$\alpha = \frac{\epsilon N_l}{\Delta t}$. Since the total circuit error rate is expressed as $\xi = \epsilon N_l \ttot$ using the
overall number of simulation steps $\ttot$, the width of the Lorentzian peak is simplified as
$\alpha  = \frac{\xi}{\ttot \Delta t}$.
Similarly, the natural linewidth due
to finite simulation time can be expressed as $\propto 1/(\ttot \Delta t)$.

As such, the broadening due to gate noise relative to the natural linewidth is simply
given by the circuit error rate $\xi$. This guarantees that in practice, i.e., when $\xi \leq 1$
the peaks are not broadened by gate noise.
This is nicely confirmed in \cref{fig:hubbard_noisy} where
the peaks for different $\xi$ have approximately the same width as given by the natural linewidth.
Furthermore, all $0$-frequency components in the ideal signal, which are ideally removed by the
standardisation of the data matrix $\dmat$, are broadened due to the exponential envelope and thus appear in the shadow spectrum 
as an intense peak at the low frequency end. Indeed, this is nicely confirmed in \cref{fig:hubbard_noisy}.

\section{Classical post-processing \label{app:post_proc}}

In this section we describe details of our classical post-processing technique and analyse its
computational complexity.  We first summarise the algorithm as illustrated in \cref{fig:summary}.
\begin{enumerate}[leftmargin=*]
	\item measure $\nobs$ time-dependent signals with classical shadows, with sample compelxity $\mathcal{O}(\log \nobs)$
	\item possibly reject signals that are not statistically significantly different from noise via a Ljung-Box test
	\item arrange the signals into a matrix $\dmat$ as row vectors
	\item calculate the square matrix $\cmat \propto \dmat^T \dmat$  -- this is the time-determining step with a complexity $\mathcal{O}(\nobs)$ 
	\item find the domianant eigenvectors $v_{1}, v_{2}, \dots  v_{\nc}$ of $\cmat$ -- we need to choose a cutoff, e.g, $c = 4$ in \cref{fig:fig1}, based on the eigenvalue distribution of $\cmat$.
	\item Finally, we effectively Fourier transform these dominant eigenvectors.
	More specifically, we obtain the spectral (Fourier) intensity at frequency $\omega$
	by calculating the spectral cross correlation matrix
	$\xmat(\omega)$ from the dominant eigenvectors $v_{1}, v_{2}, \dots  v_{\nc}$
	\item The spectrum $\lambda(\omega)$ is then obtained as the principal component of each matrix $\xmat(\omega)$
\end{enumerate}
Below we first detail the correlation analysis technique from steps 6 and 7:
This technique is closely related to generalised coherence techniques in classical signal processing~\cite{ramirez2008generalization,malekpour2018measures,santamaria2007estimation},
however, would be prohibitively expensive when directly applied to our large number of observables
$\nobs \gg \ttot$.

We thus describe a dimensionality-reduction technique in steps 2--5
which is closely related to a standard Principal Component Analysis but slightly differs as we
standardise the time signals $f_k$ (row-wise mean and standard deviation is fixed) as in \cref{eq:standardised_signals}.
In fact, this algorithm is directly analogous to well-established subspace methods in signal processing,
such as the MUSIC algorithm~\cite{hayes1996statistical}.
In these signal subspace methods one estimates a signal autocorrelation matrix, which 
	is formally equivalent to our square matrix $\cmat$ that we report in \cref{eq:cmat}.
	One then defines a signal subspace that is by definition spanned by the eigenvectors that correspond to the  largest $c$
	eigenvalues of the matrix $\cmat$. Then, the noise subspace is spanned by the complement
	of the signal subspace. In contrast to our method detailed in \cref{app:dimred}, MUSIC
	deviates as it does not use the signal eigenvectors but rather works directly
	with the noise subspace: the MUSIC spectrum is defined as  $\propto [ v(\omega)^T \tilde{\cmat} v(\omega) ]^{-1}$, 
	where $v(\omega)$ are discretised Fourier components of frequency $\omega$ and $\tilde{\cmat}$ is a projection 
	onto the noise subspace of $\cmat$. Indeed, if a frequency $v(\omega)$ is contained in the signal subspace
	then $v(\omega)^T \tilde{\cmat} v(\omega) \rightarrow 0$ and thus we observe a peak in the spectrum.

\subsection{Correlation analysis of multiple signals \label{app:cross_corr}}
We assume that we are given a set of vectors $v_1, v_2, \dots v_{\nc}$ as time signals
such that each signal, e.g., $v_1(n)$ is of total length $\ttot$ and is standardised.
We calculate the spectral cross correlation matrix between these signals as
\begin{equation}
    [\xmat(\omega_n)]_{kl} = \mathcal{F}[  X_{kl}  ](\omega_n)
\end{equation}
where we use the discrete Fourier transform, or the Fast Fourier Transform of the cross correlation between the individual signals as
\begin{equation*}
    X_{kl}(m) = \sum_{n=1}^{\ttot-m-1} v_k(n+m) v_l(n).
\end{equation*}
Here $X_{kl}(m)$ is a time signal with the temporal index $m$ and it quantifies the cross
correlation between the signals $v_k(n)$ and $v_l(n)$ as the overlap between the signals after
a time lag of $m$.
If the signals were infinitely long than the matrix elements are simply products of the
individual Fourier transforms as
\begin{equation}
    [\xmat(\omega_n)]_{kl} \propto [\mathcal{F}[ v_k  ](\omega_n)]^{*} \times \mathcal{F}[ v_l  ](\omega_n)
\end{equation}
and the matrix $\xmat(\omega_n)$ would be rank one with only 1 non-zero singular value as $\lambda(\omega_n) = \sum_{k} |\mathcal{F}[v_k](\omega_n)|^2$ which is just the squared sum of the individual spectra from \cref{eq:pds}. 

On the other hand, if the data is noisy and finite then the matrix  $\xmat(\omega_n)$ is not simply rank 1 and thus
by calculating the dominant singular values at each frequency $\lambda(\omega_n)$ we aim to estimate the
spectral density of the ideal noise free component.

The computational complexity of this approach is as follows. Calculating the cross correlations between the
signals takes $\mathcal{O}(\nc^2 \ttot)$ time  while calculating the matrices $\xmat(\omega_n)$ through
Fast Fourier Transforms takes $\mathcal{O}(\nc^2 \ttot \log\ttot )$. Finally, obtaining the
dominant singular values takes in the worst case scales as $\mathcal{O}(\nc^3 \ttot)$. As such, the worst-case scaling of the
entire procedure is $\mathcal{O}(\nc^3 \ttot)$.

Note that we could directly apply this procedure to our full dataset by considering $v_k = f_k$ from \cref{eq:standardised_signals}.
However, the computational complexity via $\nc = \nobs$ would then be $\mathcal{O}(\nobs^3 \ttot)$ which is prohibitive
in practice for a large number of observables $\nobs \gg \ttot$. For this reason we first apply a dimensionality reduction to our
data and consider only a few dominant eigenvectors as $v_1, v_2, \dots v_{\nc}$.

\subsection{Dimensionality reduction\label{app:dimred}}

In this section we assume that we have performed a time evolution in $\ttot$ time increments and recorded classical shadows at each time
increment, each containing $\nshot$ snapshots as summarised in \cref{statement:shadow}.
The classical computational complexity of reconstructing
$\nobs$ observables from the shadow data at each time increment is $\mathcal{O}(\nobs \nshot)$ using
the algorithm described in \cite{classical_shadows}.

\subsubsection{Calculating the square matrix}
We order the time-dependent expected values into the data matrix $\dmat \in \mathds{R}^{\nobs \times \ttot}$ as defined
in \cref{eq:standardised_signals}.
We can derive the time complexity of computing the square matrix $\cmat  \in \mathds{R}^{\ttot \times \ttot}$
which we obtain via the matrix-matrix product $\dmat^T \dmat$
as
\begin{equation}\label{eq:cmat}
	[\cmat]_{mn} = \frac{1}{\nobs} \sum_{k=1}^{\nobs} [\dmat]_{mk} [\dmat]_{kn} = \frac{1}{\nobs} \sum_{k=1}^{\nobs} f_k(m) f_k(n).
\end{equation}
As such, a naive implementation of calculating the matrix product has a complexity $\mathcal{O}(\ttot^2 \nobs)$.

\subsubsection{Calculating eigenvectors of the square matrix}
Imagine a time signal $v(n) \in \mathds{R}^{\ttot}$ for which we calculate the expression
with respect to the square matrix $\cmat$ as 
\begin{align*}
	\nobs \,  v^T \cmat v =& \sum_{m,n = 1}^{\ttot} v(m) C_{mn} v(n) \\
	=& \sum_{m,n = 1}^{\ttot} \sum_{k=1}^{N_o} f_k(m) f_k(n) v(m)  v(n)\\
	=& \sum_{k=1}^{\nobs}  [\sum_{m = 1}^{\ttot} f_k(m) v(m)]  [\sum_{n = 1}^{\ttot} f_k(n)  v(n)]
\end{align*}
Let us introduce the notation $\langle f_k , v \rangle := \sum_{n = 1}^{\ttot} f_k(n)  v_n$
for
scalar products between the $k$-th observable's signal $f_k$ and the probe signal $v$.
We finally obtain the expression
\begin{equation}
	 v^T \cmat v =  \frac{1}{\nobs}\sum_{k=1}^{\nobs} |\langle f, v \rangle_k|^2,
\end{equation}
as the average squared overlap between the individual signals and the probe signal $v(n)$.
As such, the average squared overlap is maximised by the dominant eigenvector $v_1$ of
the real, symmetric matrix $\cmat$ as
\begin{equation*}
	v_1^T \cmat v_1 =
	\max_{\lVert v\rVert_2=1} v^T \cmat v
	=
		\max_{\lVert v\rVert_2=1}
	\frac{1}{\nobs} \sum_{k=1}^{N_o} |\langle f, v \rangle_k|^2 
\end{equation*}
We calculate the dominant eigenvectors $v_1, v_2, \dots v_{\nc}$ of $\cmat$ where we introduce a cutoff
based on eigenvalues of the matrix $\cmat$.
The complexity of an exact diagonalization is $\mathcal{O}(\ttot^3)$ which time is independent
of the number of observables -- it only depends on the number of timesteps $\ttot$.
Finally, the complexity of dimensionality reduction is $\mathcal{O}(\ttot^3 + \ttot^2 \nobs)$. 
In numerical simulations we estimate the absolute time required to perform the entire dimensionality
reduction algorithm in \cref{fig:scaling}(c)
and find good agreement with the theoretical scaling.

\begin{figure*}[t]
	\begin{centering}
		\includegraphics[width=\linewidth]{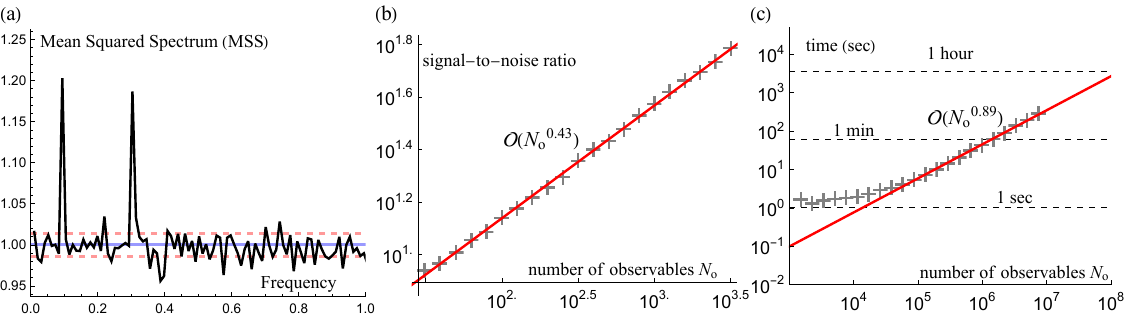}
	\end{centering}
	\caption{
		(a) same spectrum as in \cref{fig:fig1} but showing the result of a simple post-processing 
		technique whereby we calculate the sum of squares of the individual spectra as described in \cref{app:spectral_density}.
		(blue, solid) theoretical expected value $1.00$ of the baseline and (red, dashed) theoretical 
		standard deviation of the baseline $0.0137$ match well the empirical sample mean of the baseline $0.994$ and standard deviation $0.0139$.
		(b) the signal-to-noise ratio in the spectrum---the gap between the baseline and the peak relative to the standard
		deviation of the baseline---is improved as $\propto \sqrt{\nobs}$ as described in
		\cref{app:spectral_density}.
		(c) the computation time of the post-processing step is asymptotically linear in the number of observables
		while deviation from linear is expected at low $\nobs$ where calculating the eigenvectors of the constant-sized matrix $\cmat$ dominates the time.
		The absolute time is quite reasonable for even a very large number of observables, i.e., less than an hour for $\nobs \propto 10^8$.
		Computed on a desktop PC at a fixed $\ttot = 1000$ and increasing $\nobs$.
		\label{fig:scaling}
	}
\end{figure*}

\section{Shot noise propagation\label{app:shot_noise}}

In this section we derive simple analytical models to understand
how shot noise affects our shadow spectra. We then verify in numerical experiments
using actual shadow spectra that indeed these analytical models very well explain 
the effect of shot noise on shadow spectra.

\subsection{Average spectral densities \label{app:spectral_density}}
Let us consider a collection of $\nobs$ time-dependent, signals 
that we determine using classical shadows as $f_i(n) = \langle \psi(t_n) | P_i |\psi(t_n) \rangle $
with the discrete temporal index $n \leq \ttot$ and the time variable $t_n = n \Delta t$.
The simplest possible way to estimate common frequencies in the signals is by estimating the mean squared spectrum
\begin{equation}\label{eq:pds}
	\mathrm{MSS}[\omega_n] = 
	\frac{1}{\nobs}	\sum_{i=1}^{\nobs}  |F_i(n)|^2,
\end{equation}
as the average of the square of the fast Fourier transform 
$F_i(n) := \mathcal{F}[ f_i ](\omega_n)$ of each signal.

For a simple analytical model we precisely state the signal-to-noise ratio:
In \cref{eq:evol} we derived a general form of the signals. For ease of notation
we assume the system has one transition (e.g., $k=1$ and $l=2$)
with a corresponding frequency $\nu = E_k - E_l$ and intensity $I_i  = |I_{ikl}|  = |I_{i,1,2}| $.
We showed in \cref{eq:evol} that the $i^{th}$ observable gives rise to a signal
that is a phase-shifted sinusoidal function of time 
\begin{equation}\label{eq:simple_signal}
	f_i(n) = I_i \cos(\nu t_n + \phi_i),
\end{equation}
where both 
the intensity and the phase shift depends on the observable's index $i$.
Ideally, when there is no shot noise present, the spectral density
in \cref{eq:pds} yields a single peak centred at $\omega_{peak} = \nu$.
The corresponding peak intensity $\mathrm{MSS}[\omega_{peak}]  \propto  \ttot \overline{I^2}$
is given by the average $\overline{I^2} :=  \tfrac{1}{\nobs} \sum_{i=1}^{\nobs}   I_i^2$.

We now model the case when the signal has additive shot noise due to a finite number of classical snapshots.
For a reasonably large number of shots the shot noise can be very well approximated by a normally distributed
random variable.
It is well known that the discrete Fourier transform of normally distributed noise
is normally distributed white noise.
Thus, to a good approximation, 
every point in the discrete Fourier transform is a sum $F_i(n) + \mathcal{E}_i(n)$
of the ideal value $F_i(n) $
and an additive random variable $\mathcal{E}_i(n)$ -- 
these random variables are independent and for ease of notation we assume they have an identical variance which we denote
as
$\var[\mathcal{E}_i(n)] = \epsilon^2$ and the mean is computed as $\mathbb{E}[\mathcal{E}_i(n)] = 0$.
The expected value of the mean squared spectrum at the peak centre can be calculated as 
\begin{align*}
	\mathbb{E} ( \mathrm{MSS}[\omega_{peak}] ) 
	&= \frac{1}{\nobs} \sum_{i=1}^{\nobs} \mathbb{E}[  (\ttot I_i +  \mathcal{E}_i(\omega_{peak}))^2]\\
	&= \ttot \overline{I^2}  +  \epsilon^2.
\end{align*}
We can also calculate the mean value of the baseline of the spectrum
by focusing on a point in the spectrum $\omega_{base}$ that is far away from the peak centre $\omega_{peak}$,
where the ideal spectrum would be zero, thus
\begin{equation*}
	\mathbb{E} ( \mathrm{MSS}[\omega_{base}] ) 
	= \frac{1}{\nobs} \sum_{i=1}^{\nobs} 	\mathbb{E} ( \mathcal{E}_i(\omega_{base})^2 )
	= \epsilon^2.
\end{equation*}
As such, the baseline of the spectrum is not zero but is actually a constant that is determined by the level of shot noise $\epsilon$.
Since the baseline is constant, the signal-to-noise ratio is determined 
by the standard deviation of the baseline which we can calculate via the variance as
\begin{equation*}
	\var ( \mathrm{MSS}[\omega_{base}] ) 
	= \frac{1}{\nobs^2} \sum_{i=1}^{\nobs}  \var[ \mathcal{E}_i(\omega_{base})^2]
	=  \frac{2 \epsilon^4}{\nobs},
\end{equation*}
where we used that the variance of the square of a random variable of mean zero is $\var[ \mathcal{E}_i(\nu')^2] = 2 \epsilon^4$.
It follows that the standard deviation of the baseline is $\epsilon^2\sqrt{2 /\nobs}$, 
and thus the signal to noise ratio can be calculated as
\begin{equation} \label{eq:SNR}
	\mathrm{SNR} = \frac{\ttot \overline{I^2} \sqrt{\nobs} }{ \sqrt{2} \epsilon^2}
	\propto
	\ttot \nshot \overline{I^2} \sqrt{\nobs},
\end{equation}
where in the last equation we use that the statistical fluctuations are due to the finite number of circuit
repetitions (shot noise) as
$\epsilon^{-2} \propto \nshot$. As such, in order to resolve a peak we need to
set the total shot budget $\ttot \nshot$ as the amount of quantum resources such 
that $1/\overline{I^2} < \ttot \nshot \sqrt{\nobs}$.
Indeed, the signal-to-noise ratio can be improved just by increasing $\sqrt{\nobs}$ at a fixed 
shot budged $\ttot \nshot$, i.e., without using more quantum resources. Furthermore, 
using the pre-screening described in \cref{sec:postproc} one can discard observables
with negligible signal guaranteeing that  $\overline{I^2} > I_{min}^2 > 0$, while
the actual value of $\overline{I^2}$ depends strongly on the considered system.

We find that this simple model quite accurately describes the effect of shot noise on actual shadow spectra:
in \cref{fig:scaling}(a) we calculate the average spectral density for the same simulation
as in \cref{fig:fig1}. As we describe in \cref{sec:postproc}, we standardise our data matrix such
that the standard deviation of each signal is unity and thus the variance in the Fourier
transform is approximately $\var[\mathcal{E}_i(n)] \approx \epsilon^2 \approx 1$. Indeed, 
\cref{fig:scaling}(a) confirms that the expected baseline $\epsilon^2 =1$ (blue solid line) in the spectrum matches well the
empirical sample mean $0.994$.
\cref{fig:scaling}(a) also confirms that the theoretically expected standard deviation
of this baseline $\sqrt{2/\nobs} \epsilon^2 = \sqrt{2/10689} \approx 0.0137$ nicely matches the empirical sample standard deviation $0.0139$.

In \cref{fig:scaling}(b) we randomly select $\nobs$ observables from
a pool of all at most 3-local Pauli strings with $\tilde{\nobs} = 10689$.
The expectation value of the
average signal intensity for such a random selections is $\mathbb{E} [\overline{I^2}] =  \tfrac{1}{\tilde{\nobs}} \sum_{i=1}^{\tilde{\nobs}}   I_i^2$
the average over the full set which is a constant and is independent of the number of selected terms $\nobs < \tilde{\nobs}$. 
\cref{fig:scaling}(b) shows the average of the signal-to-noise ratio of $100$ random selections 
of $\nobs$ observables for an increasing $\nobs$.
Indeed, \cref{fig:scaling}(b) confirms that the signal-to-noise ratio increases
as $\propto \sqrt{\nobs}$. Note that such a random selection of observables can be considered as
a worst-case scenario when no information about the observables is assumed.
Of course, in the present work we explore various advanced techniques to improve upon
	randomly selecting observables, e.g., the pre-screening step described in \cref{sec:postproc} 
can be used to identify the $\nobs$ most intense signals in the pool of total size $\tilde{\nobs}$
so that a well-informed decision can be made.

\subsection{Dimensionality reduction of periodic signals\label{app:dimred_per}}
Here we build a simple analytical model to understand how shot noise affects
the dimensionality reduction approach introduced in \cref{sec:postproc}.
First, let us write the phase-shifted cosine function
in \cref{eq:simple_signal} as a linear combination
of cosine and sine functions as
\begin{equation*}
	f_i(n) = I_i \cos(\nu t_n + \phi_i) =  c_i \cos(t_n)  + s_i \sin(t_n),
\end{equation*}
with amplitudes $c_i = I_i \cos(\phi_i)$ and $s_i = - I_i \sin(\phi_i)$.
Here we again used the discrete temporal index $n \leq \ttot$ and the time variable $t_n = n \Delta t$.
Since discretised vectors of cosine and sine functions are mutually orthogonal,
we can write the signal as a 2-dimensional vector in Fourier basis
with entries as $\tilde{f}_i(t) = 	 \sqrt{\ttot} ( c_i ,  s_i)$ where $\sqrt{\ttot}$ ensures normalisation of the basis vectors.
Thus, each observable's signal can be represented by a 2-dimensional vector (in the Fourier basis)
and  arranging all $\nobs$ signals into a data matrix
yields the effective $2 \times \nobs$-dimensional matrix $\tilde{\dmat}$,
where the tilde refers to the fact that we represent the matrix in a Fourier basis, as
\begin{equation*}
	 \tilde{\dmat} = 
	 \sqrt{\ttot}
\begin{bmatrix}
		c_1 & s_1 \\
		c_2 & s_2 \\
\vdots & \vdots\\
	c_{\nobs} & s_{\nobs}
\end{bmatrix},
\quad \quad
\tilde{\cmat} = 
\frac{\ttot}{\nobs}	\begin{bmatrix}
		\sum_i c_i^2 & \sum_i c_i s_i \\
		\sum_i c_i s_i & \sum_i s_i^2 
	\end{bmatrix} .
\end{equation*}
Above we have also calculated the $2\times2$ square matrix $\cmat$.
We can readily evaluate the sum of the two eigenvalues of this matrix as
 $ c_i^2 {+} s_i^2 = \ttot \overline{I^2} $
which is identical to the peak height in the mean squared spectrum in \cref{eq:pds}
where $\overline{I^2}$ is  again the average peak intensity as $\overline{I^2} :=  \tfrac{1}{\nobs} \sum_{i=1}^{\nobs}   I_i^2$. 
Since the eigenvalues are invariant under a unitary transformation, it
immediately follows that the matrix $\cmat$ (without tilde, not in Fourier basis) has the same
two non-zero eigenvalues.

Similarly to \cref{app:spectral_density}, we now assume that under the effect of shot noise
each signal is a sum $f_i(n) + \mathcal{E}_i(n) $ of an ideal component $f_i(n)$
and an additive random variable $\mathcal{E}_i(n)$. For ease of notation
we assume the random variables are independent and have the same variance as
$\var[\mathcal{E}_i(n)] = \epsilon^2$ and mean $\mathbb{E}[\mathcal{E}_i(n)] = 0$.
Given the data matrix entries as $d_{mi} := f_i(m) + \mathcal{E}_i(m)$, we can
calculate the noisy matrix $\cmat_{noisy}$ as
\begin{equation*}
	[\cmat_{noisy}]_{mn} = \frac{1}{\nobs}\sum_{i=1}^{\nobs} d_{mi} d_{in}.
\end{equation*}
We can calculate the mean of the noisy matrix and find that only the diagonal
entries are affected by shot noise as
$\mathbb{E}[ 	\cmat_{noisy} ]  =  \cmat + \epsilon^2 \mathrm{Id}$,
resulting in a sum of the ideal matrix and a matrix that is proportional to the identity.
Thus the random noise has no effect on the expected value of the eigenvectors
and it just shifts the ideal dominant eigenvalues by a constant factor as $\ttot \overline{I^2} + \epsilon^2$ in direct
analogy with \cref{app:spectral_density}.

As such, in analogy with \cref{app:spectral_density} the signal-to-noise ratio is determined by the standard deviation
of the matrix entries.
We can straightforwardly evaluate the variance of the off-diagonal entries $m \neq n$  as
\begin{align*}
	\var[ 	(\cmat)_{mn}]  &= \var[   \frac{1}{\nobs} \sum_{i=1}^{\nobs} d_{mi} d_{in} ]\\
	 &=  \frac{1}{\nobs^2} \sum_{i=1}^{\nobs} \var[ d_{mi} d_{in} ] \\
	&= \frac{1}{\nobs^2}  \sum_{i=1}^{\nobs}  \var[ d_{mi}] \var[ d_{in} ]  =  \frac{ \epsilon^4}{\nobs}.
\end{align*}
The variance of the diagonal entries can be calculated similarly and yields the same expression $\frac{ \epsilon^4}{\nobs}$.

In summary, we considered a system with only one frequency component as in \cref{eq:simple_signal}
and evaluated analytically that the matrix $\cmat$ has only 2 nonzero eigenvalues and the sum of eigenvalues is
$\ttot \overline{I^2}$. Furthermore, the standard deviation of the matrix entries due to shot noise 
is $\epsilon^2/\sqrt{\nobs}$ which shifts the non-zero eigenvalues. Thus in order to tell apart
the dominant eigenvalues from the noisy ones there must be a sufficient gap between $\ttot \overline{I^2}$
and $\epsilon^2/\sqrt{\nobs}$. This is effectively the same conclusion we found in
\cref{app:spectral_density} where we used a simple Mean Square Spectrum estimate.
As such, while the dimensionality reduction step is not expected to decrease noise in the
signals, it allows us to make use of advanced spectral cross correlation methods described in \cref{sec:postproc} --
where the latter can suppress noise by a constant factor---compare \cref{fig:scaling}(a) with
the spectrum in \cref{fig:fig1}---but we do not expect it can fundamentally outperform the
simple mean squared spectrum estimate in \cref{app:spectral_density}.

\begin{figure*}[t]
	\begin{centering}
		\includegraphics[width=\linewidth]{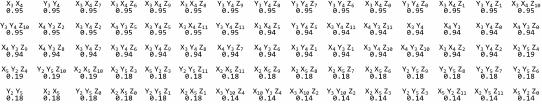}
	\end{centering}
	\caption{
		List of the 75 at most $3$-local Pauli strings, which give most intense signals. 
		  The intensities are given by the matrix elements $|\langle \psi_{g} | O_i | \psi_{e} \rangle|$,
		where $| \psi_{g} \rangle$ and $| \psi_{e} \rangle$ are the exact ground and the first triplet excited state
		of a 12-qubit LiH Hamiltonian. The most intense observables are the ones that correspond
		to the HOMO-LUMO triplet transition of the basis states $|111100\dots\rangle \rightarrow |111010\dots\rangle$
		which include, e.g,  $X_3 X_4$, $X_3Y_4$ or any Pauli string obtained by adding a Pauli $Z$ operator on any other qubit,
		e.g., $Y_3 Y_4 Z_7$.
		The next most intense ones
		  also generate HOMO-LUMO transitions to the triplet excited state $|111100\dots\rangle \rightarrow |110101\dots\rangle$
		incuding e.g. $X_2 Y_5$ or any other string that is related to this by adding a Pauli $Z$ operator to any other qubit.
  }
		\label{table:HF_paulis}
\end{figure*}

\section{Observables in Quantum Chemistry\label{app:qchem}}
\subsection{Fermionic operators}
We now analyse the example of a peak that corresponds to the gap between the ground state
and the first excited state.
The peak intensity  with respect to an observable $O$ (\cref{statement:intenisty}) is calculated via an expression
formally resembling the conventional transition dipole moment as
\begin{equation} \label{eq:expansion}
	I_{g\rightarrow e} \propto \bra{\psi_g} O \ket{\psi_e}.
\end{equation}
Here we write both the exact ground and excited states as a Configuration Interaction (CI) expansion, i.e.,
a sum of the HF Slater determinant and post-HF corrections:
\begin{align*}
    \ket{\psi_s} &= c^{(s)}_0 \ket{\text{HF}} + \sum_{pq} c^{(s)}_{pq} a^\dagger_p a_q\ket{\text{HF}} \\
    &\quad + \sum_{pqrs} c^{(s)}_{pqrs} a^\dagger_pa^\dagger_q a_ra_s\ket{\text{HF}} + \dots,
\end{align*}
where $\ket{\text{HF}}$ is the mean-field Hartree-Fock state~\footnote{Typically the HF solution has the dominant weight
as $|c^{(g)}_0| > |c^{(g)}_{pq}|$  (even if the HF Slater determinant is not a good approximation)
and thus in \cref{eq:expansion} the first term in the expansion is the dominant term.}.
Thus \cref{eq:expansion} can be expanded as
\begin{align}\label{eq:hf_expansion}
    \bra{\psi_g} O \ket{\psi_e} = & c^{(g)*}_0 c^{(e)}_0 \bra{\text{HF}}O\ket{\text{HF}} \\ \nonumber
    & + c_0^{(g)*}\sum_{pq}c^{(e)}_{pq}\bra{\text{HF}} Oa^\dagger_{p}a_{q}\ket{\text{HF}}\\ \nonumber
    &+ c_0^{(e)}\sum_{pq}c^{(g)*}_{pq}\bra{\text{HF}} a^\dagger_{q} a_{p}  O\ket{\text{HF}}
    + \dots
\end{align}
Above we only expand up to the single-excitation CI terms
from which it is clear that if $O$ cancels out the excitation operators,
we get $\bra{\text{HF}}\text{HF}\rangle=1$ and the corresponding CI coefficients would contribute to $\bra{\psi_g} O \ket{\psi_e}$. 

In particular, choosing the observable as the Hermitian operator that corresponds to the single excitation
$O= a^\dagger_{m} a_n {+}  a^\dagger_n a_{m}$ with $m$ acting on occupied and $n$ acting on virtual orbitals,
the leading terms in \cref{eq:hf_expansion} can be evaluated analytically as
\begin{align}\label{eq:fermionic_exp}
	\bra{\text{HF}} O\ket{\text{HF}}
	&=
	\bra{\text{HF}} a^\dagger_{m} a_n {+}  a^\dagger_n a_{m} \ket{\text{HF}}= 0,
	\\ \nonumber
	\bra{\text{HF}} O a^\dagger_{p}a_{q}\ket{\text{HF}}
	&=
	\bra{\text{HF}} a^\dagger_{m} a_n a^\dagger_{p}a_{q}\ket{\text{HF}} {+}0
	=\delta_{np} \delta_{mq},\\ \nonumber
	\bra{\text{HF}}a_{q}^\dagger a_{p} O\ket{\text{HF}}
	&= 0{+}\bra{\text{HF}}  a^\dagger_{q} a_{p} a^\dagger_n a_{m} \ket{\text{HF}}
	=\delta_{np} \delta_{mq}.
\end{align}
Thus, in leading order, the intensity of the observable
$\bra{\psi_g} a^\dagger_{m} a_n {+}  a^\dagger_n a_{m} \ket{\psi_e}$ 
is determined by the corresponding single-excitation coefficients $c_0^{(g)*} c^{(e)}_{mn}$
and $c_0^{(e)} c^{(g)*}_{mn}$ as all other terms shown explicitly in \cref{eq:hf_expansion}
give null contributions. Of course higher order CI terms in 
\cref{eq:hf_expansion} might also contribute, but they might have lower contributions.

Furthermore, by choosing an observable of a higher excitation, the intensity of the corresponding
signal will be determined in leading order by its higher order CI coefficients
in \cref{eq:hf_expansion}; for example, by choosing an observable
$O=a^\dagger_pa^\dagger_q a_ra_s{+} a^\dagger_s a^\dagger_r a_q  a_p$, the intensity
in leading order is determined by the corresponding expansion coefficient $c^{(s)}_{pqrs}$.
Thus, by simultaneously estimating expected values of all $q$-local fermionic operators one
obtains a large number of signals but the intense signals will correspond to excitations
that have dominant expansion coefficients in \cref{eq:hf_expansion}. Such fermionic observables
can be estimated using near-term-friendly fermionic
shadows~\cite{wanMatchgateShadowsFermionic2022, zhaoFermionicPartialTomography2021}.

\subsection{Jordan-Wigner Encoding}\label{stat:jw_enconding}
By inspecting the JW transform of fermionic excitation operators, it is straightforward
to see how a pool of $q$-local Pauli operators, even for $q\sim 4$, is sufficient to capture
single electron transitions between meaningful eigenstates, which is ideal for our classical-shadow-based approach.
\begin{statement}
Choosing the Pauli bitflip observable $O = X_m X_n$, the signal intensity $\bra{\psi_g} X_m X_n \ket{\psi_e}$
is determined in leading order by the single-excitation expansion coefficients
$c_0^{(g)*} c^{(e)}_{mn}$ and $c_0^{(e)} c^{(g)*}_{mn}$ of the state  from \cref{eq:fermionic_exp}.
Similarly, if we choose any other operator $O(m,n)$ by replacing any of the two Pauli $X$ operators in $O$ with Pauli Y operators
and we additionally append Pauli $Z$ operators on any qubit other than $m$ or $n$
then the intensity is still determined by the same expansion coefficient. 
\end{statement}
\begin{proof}
In general, the excitation operators in the JW encoding are mapped to
\begin{align*}
	&a_p^\dagger \mapsto \alpha^\dagger_p:= Z_{:p} A^\dagger_p,  \quad \quad a_p \mapsto \alpha_p:= Z_{:p} A_p\\
	& a_p^\dagger a_q \mapsto i A_p^\dagger Z_{q:p} A_q
\end{align*}
where $Z_{:p} := \prod_{k=1}^{p-1} Z_k$ and $A = (X + i Y)/2$ is the qubit lowering operator
and  $Z_{q:p} = \prod_{k=p+1}^{q-1} Z_k$.
Furthermore, the HF state becomes a computational basis state $\ket{\text{HF}}\mapsto\ket{11\dots1100\dots} \equiv \ket{b}$ where $b$ is simply a binary number.

Thus for any Pauli string $O = \bigotimes_{k=1}^N P_{j_k}$ the expected value
in \cref{eq:fermionic_exp} splits up as a product of single-qubit
expected values as
\begin{align*}
	&\bra{b} O(m,n) \alpha^\dagger_{p}  \alpha_{q}\ket{b} 
	=	\bra{b} O(m,n)  A_p^\dagger Z_{q:p} A_q   \ket{b} \\
	&= \bra{b_p} P_{j_p}  A^\dagger  \ket{b_p} \bra{b_q} P_{j_q}  A    \ket{b_q}  
	\prod_{k \neq p, k \neq q} \bra{b_k} P_{j_k} Q_{pq}(k)   \ket{b_k},
\end{align*}
where $p_p$ is the $p^{th}$ bit of the binary number $b$ and $Q_{pk}(k) = Z$ if $q < k < p$ and the identity operator otherwise.
We can evaluate the product of coefficients as
\begin{equation*}
	|\bra{b_p} P_{j_p}  A^\dagger  \ket{b_p} \bra{b_q} P_{j_q}  A    \ket{b_q}  | = \delta_{np} \delta_{mq}/4,
\end{equation*}
assuming that $m<n$ without loss of generality
and we only show the absolute value of the expression for ease of notation.
Above the expression evaluates in absolute value to $1$ only when the two Pauli $X$ (or $Y$) operators
in the observable $O(m,n)$ act on the same two qubits to which $A$ and $A^\dagger$ act, i.e., to qubits
$p$ and $q$. Furthermore, $	\prod_{k \neq p, k \neq q} \bra{b_k} P_{j_k} Q_{pq}(k)   \ket{b_k} \in \{\pm 1\}$ is only a sign
factor that is determined by the bitsring $b$ and by the number of Pauli $Z$
operators in $P$ and in $Z_{p:q}$. 
Thus we conclude that the expected values from \cref{eq:fermionic_exp} in the JW encoding
evaluate to
\begin{align}
	\bra{b} O(m,n) \ket{b} = 0,
	\\ \nonumber
	|\bra{b} O(m,n) \alpha^\dagger_{p} \alpha_{q}\ket{b}|
	&= \delta_{np} \delta_{mq}/4 \\ \nonumber
	|\bra{b} \alpha^\dagger_{q} \alpha_{p} O(m,n)\ket{b}|
	&= \delta_{np} \delta_{mq} /4.
\end{align}
The signal intensity 
$\bra{\psi_g} O(m,n) \ket{\psi_e}$ of the observable $O(m,n)$
is determined in leading order by the single-excitation expansion coefficients
$c_0^{(g)*} c^{(e)}_{mn}$ and $c_0^{(e)} c^{(g)*}_{mn}$ of the state  from \cref{eq:fermionic_exp}.
\end{proof}

The crucial observation we make in the above statement is that even though fermionic excitation operators
$a_p^\dagger a_q \mapsto i A_p^\dagger Z_{q:p} A_q$ are mapped to non-local Pauli strings in the JW encoding, we can
still resolve signals of the corresponding operators just by measuring 2-local Pauli strings $X_p X_q$ or any variant
where we additionally append Pauli Z operators.
\cref{table:HF_paulis} summarises the at most $3$-local Pauli string signals for a 12-qubit LiH Hamiltonian that give the largest $\bra{\psi_g} O \ket{\psi_e}$ values, and operators like e.g. $X_3 X_4$ indeed give the most intense signals. 
The above statement also suggest that 3-local Pauli strings like $X_3 Y_4 Z_{7}$, where the $Z$ operator is placed on any site, can still give rise to strong signals.

The information above can be used to inform which specific Pauli string observable to reconstruct.
Furthermore, we could use e.g. classical MP2 simulations to estimate which excitations will be dominant and thus construct the corresponding families of Pauli strings.
If we know which operators to measure, they can in principle be ordered into a small number
of commuting groups and we could thus measure them simultaneously using well-established techniques e.g. Ref.~\cite{Izmaylov2020}.
Since we can efficiently estimate the signal of every $q$-local Pauli string using classical shadows (as the data matrix $\dmat$) and
perform a simple statistical autocorrelation test on the signals to determine whether they contain
oscillations, we need not perform any \textit{a priori} prediction on expected operator contributions. 
We did this on the present LiH example and sorted the observable signals according to their statistical p-values; 
a virtually identical list of Pauli strings as per the exactly calculated \cref{table:HF_paulis} was obtained.

\begin{figure}[tb]
	\centering
	\includegraphics[width=0.45\textwidth]{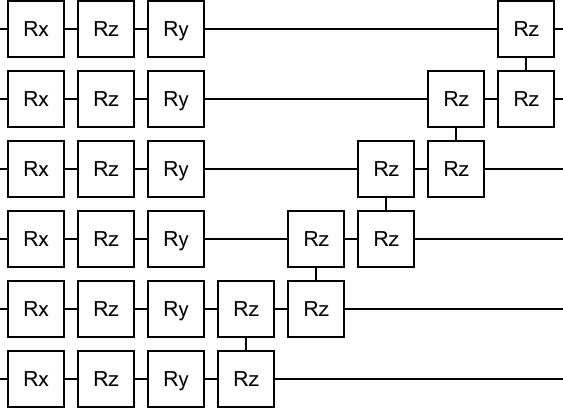}
	\caption{A single layer of the hardware-efficient ansatz used for variational dynamics of the disordered Heisenberg chain.}
	\label{fig:hardware_efficient_ansatz}
\end{figure}

\section{Additional details of numerical simulations}
\subsection{Ansatz-based variational quantum simulation\label{app:variationalsimulation}}
Our ansatz-based variational quantum simulations follow the methods introduced in Ref.~\cite{ying_li}. At each time-step, we calculate the quantum metric tensor
\begin{equation}
    A_{ij}=\frac{\partial\bra{\psi(\bm{\theta})}}{\partial \theta_i}\frac{\partial\ket{\psi(\bm{\theta})}}{\partial\theta_j}
    \label{eqn:metrictensor}
\end{equation}
and gradient vector
\begin{equation}
    C_i=\frac{\partial\bra{\psi(\bm{\theta})}}{\partial\theta_i}\hat{H}\ket{\psi(\bm{\theta})}
    \label{eqn:gradientvector}
\end{equation}
from the variational state. We then solve a corresponding linear system
\begin{equation}
    \sum_j \operatorname{Re}(A_{ij})\dot{\theta}_j=\operatorname{Im}(C_i)
\end{equation}
to find the parameter time derivatives $\dot{\theta}_i$ such that we can compute the updated parameters as $\theta'_i=\theta_i+\dot{\theta}_i\delta t $. In practice, the metric tensor $A$ may be ill-conditioned due to degeneracies in parameter space, and so we apply Tikhonov regularization by instead finding $\dot{\bm{\theta}}$ that minimizes
\begin{equation}
    \|\bm{C}-\bm{A}\dot{\bm{\theta}}\|^2+\lambda\|\dot{\bm{\theta}}\|^2.
\end{equation}
We find that a regularization hyperparameter $\lambda=10^{-4}$ and a timestep $\delta t=10^{-2}$ is sufficient to keep our fidelities with respect to the true dynamical state above 99\%.

A Wick-rotated form of these dynamics can be used to variationally prepare ground states \cite{mcardle2019variational,PhysRevA.106.062416,PRXQuantum.2.030324}. In variational imaginary time evolution, the linear system at each timestep is instead
\begin{equation}
    \sum_j \operatorname{Re}(A_{ij})\dot{\theta}_j=-\operatorname{Re}(C_i).
\end{equation}
We apply this method to prepare our initial variational state, intentionally terminating before convergence to ensure non-negligible overlap with low-lying excited states.

Note that the given forms of the metric tensor (\cref{eqn:metrictensor}) and gradient vector (\cref{eqn:gradientvector}) are based on McLachlan's variational principle and do not implicitly account for a possible global phase mismatch with the target state. To ensure correct evolution, we include a global phase parameter in the ansatz. This yields the same linear system at each step as would be obtained by including explicit global phase correction terms in the equations of motion \cite{yuan2019theory}. Our numerics use a 5-layer hardware-efficient ansatz of the form depicted in \cref{fig:hardware_efficient_ansatz}, with single-qubit and nearest-neighbour $ZZ$ rotations only.

\begin{figure}
	\centering
	\includegraphics[width=0.45\textwidth]{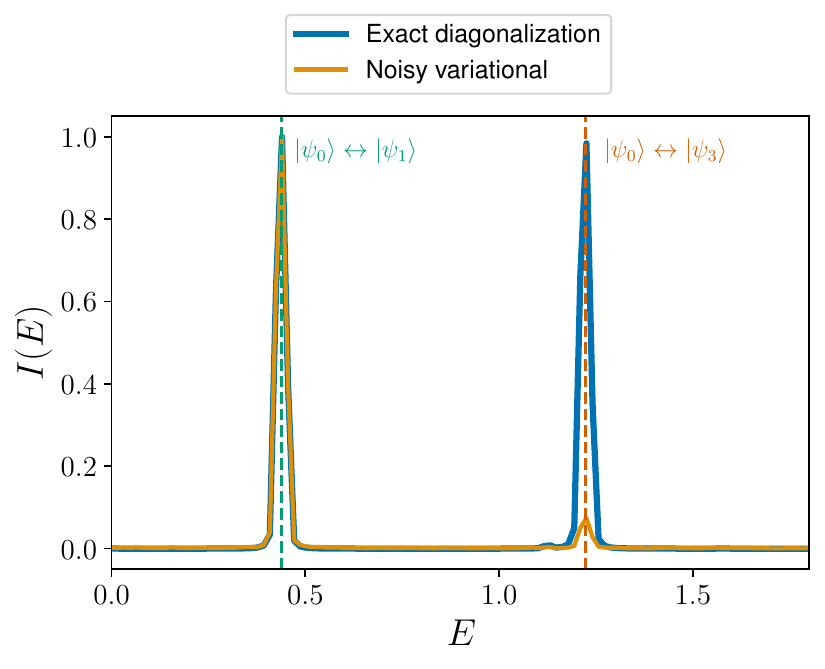}
	\caption{Shadow spectra for the disordered Heisenberg chain obtained with dynamics via exact diagonalization (blue) and simulated noisy variational evolution (orange). The initial state has non-trivial support for $\ket{\psi_3}$, resulting in a peak corresponding to the $\ket{\psi_0}\leftrightarrow\ket{\psi_3}$ transition in the ideal spectrum. However, the algorithmic error associated with the variational evolution suppresses this peak.}
	\label{fig:exact_vs_variational}
\end{figure}

In Figure \ref{fig:exact_vs_variational}, we provide additional simulations to highlight which factors affect secondary peak intensity
	and use a higher number of shots than in the hardware experiments, i.e., 4000 timesteps and 3$\times$50 shots at each timestep.
	Specifically, we compare exact time evolution (blue line) to our variational approximation (orange line). Here we include gate noise comparable to typical state-of-the art experimental devices in our variational observable calculations, applying a two-qubit depolarizing map described by
\begin{equation}
\label{eqn:twoqubit_depolarizing}
    \Phi_{i,j}^\lambda(\rho) = \left(1 - \frac{16\lambda}{15}\right) \rho + \frac{\lambda}{15}\sum_{\mathclap{\nu_1, \nu_2 \in \{1, x, y, z\}}} \sigma_i^{\nu_1} \sigma_j^{\nu_2} \rho\, \sigma_j^{\nu_2} \sigma_i^{\nu_1},
\end{equation}
where $\sigma_i^{\nu}$ a Pauli-$\nu$ operator on qubit $i$ and $\sigma_i^{1}$ is the identity. We choose error rates $\epsilon_2=10^{-3}$ after each two-qubit gate and the usual single-qubit depolarizing map with error rate $\epsilon_1=0.25\epsilon_2$. When computing the spectrum from exact diagonalization for comparison, the signal variances are \textit{not} standardized, in order to preserve relative peak heights that would otherwise be enforced by a shot-noise floor.

\cref{fig:exact_vs_variational} highlights that the secondary peak can be resolved even under the effect of gate noise
as more shots were used than in the experiment.
Furthermore, \cref{fig:exact_vs_variational} highlights the effect of approximate variational evolution on peaks with weak support. We observe that despite relatively small support in the $\ket{\psi_3}$ eigenstate, the associated ground state transition is clearly resolved in the spectrum by our post-processing method when computing observables using exact time-evolution. However, the algorithmic error associated with approximate variational evolution strongly suppresses this peak. In \cref{app: hardware}, we see that fewer shots used and stronger noise can completely eliminate these suppressed peaks.

\begin{figure*}
	\centering
	\includegraphics[width=\textwidth]{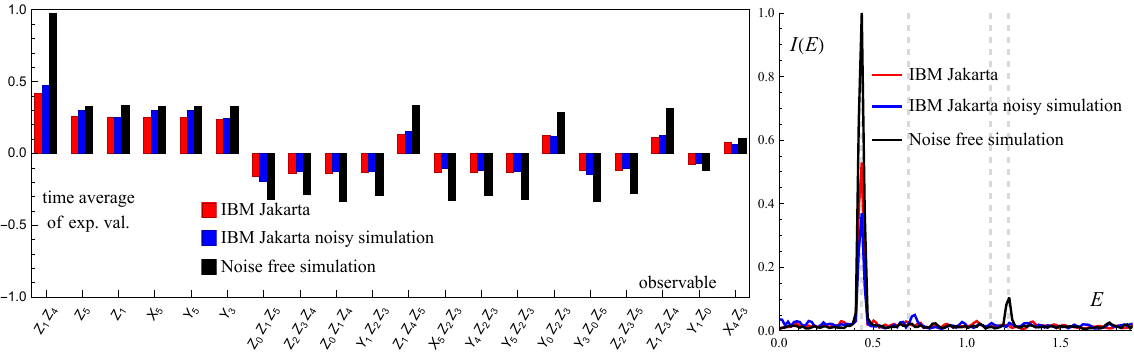}
	\caption{
			(left)
			Expected values of time-independent observables are estimated from classical shadows obtained from IBM Jakarta (red)
			and from noisy simulations of IBM Jakarta (blue).
			The level of shot noise in these estimates is low with an average standard deviation of $0.01$ thus no error bars are shown.
			These expected values are shrunk compared to the ideal ones (black) by an average factor of 
			$\eta_{exp} \approx 0.54$ in the experiment ($\eta_{sim} \approx 0.55$ in simulation),
			but with a significant variability across the different observables.
			The hardware noise is comparable to a circuit error rate
			$\xi  \approx -\ln \eta_{exp} \approx 0.62$ which is comparable but smaller than in our simulations of
			the Hubbard model in \cref{app:hubbard_noise_models}.
			(right)
			Estimating the experimental shadow spectrum using all 220 shots per timestep obtained from IBM Jakarta.
			The noisy simulation predicts slightly lower peak intensity but this deviation is well explained by the high
			variability of the shrinking of observables (see text). An additional, faint peak is visible in the noise-free simulation
			which could be resolved by using more shots.
		\label{fig:experiment}
	}
\end{figure*}

\subsection{Hubbard model\label{app:hubbard}}
The time evolution between measurements is implemented using the first order Lie-Trotter-Suzuki product formula
\begin{equation*}
    e^{-iHt} = \prod_{k = 1}^{N_\mathrm{Trott}} \prod_{\ell = 1}^{L} e^{-i H_\ell \,\delta t} + \mathcal{O}(\delta t^2)
\end{equation*}

We assume that in an early fault-tolerant scenario the dominant source of error are the applications of $T$ gates and further assume that the dominant
error mechanism is the imperfect  magic state distillation which we describe via a depolarising noise model.
We thus assume a two-qubit depolarising channel is applied for each hopping term in the Hamiltonian with probability $\lambda$
as per \cref{eqn:twoqubit_depolarizing}.
This channel is applied after every multi-Pauli rotation of the form
\begin{equation}
    \exp\left(-i \,\delta t\, \sigma_i^{\nu} \sigma_{i+1}^{z} \ldots \sigma_{j-1}^{z}\sigma_j^{\nu}\right).
    \label{eqn:pauli_gadget}
\end{equation}
with $\nu \in \{x, y, z\}$. These types of rotations appear naturally from the Jordan-Wigner transformation of hopping terms. Although not directly employed in our calculations, the $\sigma^z$ terms sandwiched between qubits $i$ and $j$ can be removed by introducing a network of Fermionic \textsc{swap} (\textsc{Fswap}) gates, which only consists of local gates of depth $\mathcal{O}(N^{\frac{1}{2}})$~\cite{PhysRevApplied.14.014059,PhysRevApplied.18.044064,PhysRevA.79.032316,PhysRevLett.120.110501}. This motivates an error model whereby depolarising noise acts on qubits $i$ and $j$ for each of these terms. Single-qubit rotations are burdened with the equivalent depolarising channel for a single qubit at the same noise strength.

\section{Additional details of hardware demonstration}
\label{app: hardware}

The shadow spectrum of the 6-site linear Heiseberg spin chain was quantum-computed on the IBM Quantum platform, using the 7-qubit IBMQ Jakarta QPU (Falcon r5.11H processor) and the QASM simulator, both freely accessible in the Open Plan -- no device reservation or queue priority privileges were used in these experiments. A set of 1000 time-evolved states at regular intervals were selected and prepared with the variational ansatz circuit and optimised classical parameters $\mathbf{\theta}_t$ that uniquely define each time-evolved state (see details in \cref{app:variationalsimulation}) -- we do not perform variational optimisation on actual hardware but use instead the parameters computed from earlier numerical simulations.

At each time step, 220 classical shadow measurements were taken; this was done by appending random basis rotations to 220 copies of an ansatz, and submitting them as one Qiskit Runtime Sampler primitive job with the option of performing only a single shot measurement for each circuit. Other settings of the Sampler jobs were taken as defaults - specifically the resilience level argument was set to 1 (minimal mitigation costs) and optimisation level set to 3 (even heavier optimisation). The 1000 Sampler primitives were further batched into 10 Qiskit Sessions, where 1 every 10 time-steps were grouped into one Session (i.e. Sampler jobs indexed 0, 10, 20,\dots, 90 and 1, 11, 21\dots, 91 etc. would correspond to the first two Sessions). This took 9 days in total to complete. We highlight that while 220 shots were measured at each time-step, for the spectrum in \cref{fig:variational_spectrum} we only actually used $3\times50$ randomly selected shots (split into 3 batches of 50 shots each for the median of
means estimation from classical shadows).

Using these 1000 time-evolved classical shadows we determined expected values of all 3-local Pauli strings
and ordered them into a standardised data matrix $\dmat$ as detailed in \cref{app:post_proc}.
We performed a Ljung-Box test and retained only the highest scoring 50 signals in $\dmat$.
The correlation matrix $\cmat$ and its eigenvectors were then computed and the 5 eigenvectors
with the highest eigenvalues were retained. The shadow spectrum was then computed from these 5
eigenvectors via the spectral cross-correlation method described in \cref{app:post_proc} (effectively computing
Fourier transforms of the eigenvectors).

\section{Analysis of the experimental data}

In addition to the experimental spectrum in \cref{fig:variational_spectrum} using only $3\times50$ shots, here we analyse
the performance of shadow spectroscopy using all $220$ snapshots per time step collected from IBM Jakarta.
These snapshot were again split into 3 batches and all up to three-local Pauli strings were estimated using
classical shadows. 

First, we analyse the noise model of the experimental device and compare results to
noisy simulations using the theoretical noise model of IBM Jakarta by passing the FakeJakarta backend into the options of each Sampler Job run on IBM's QASM simulator (total compute time of $<0.5$ hour). Recall that under global depolarising noise at a circuit error rate 
$\xi$ every Pauli string's expected value is shrunk by the same factor $\eta_{exp} \approx e^{-\xi}$ resulting
in the global relation between experimental and ideal expected values
$\langle P \rangle_{exp} = \eta_{exp} \langle P \rangle_{id}$.
In \cref{fig:experiment}(left) we plot observable expected values and compare them to ideal and
noisy simulations of IBM Jakarta: indeed
the noisy observed expected values are shrunk compared to the ideal ones but not uniformly,
i.e., some observables are shrunk more than others. This indeed confirms expectations
that in the case of variational Hamiltonian ansaetze the global depolarising model is only
a crude approximation~\cite{whitenoise}. We estimate that the observables on average are shrunk 
by a factor $\eta_{exp} \approx 0.54$ in the experiment and by $\eta_{sim} \approx 0.55$ in simulation,
but with a significant variability across the different observables via the standard deviations $\sigma_{exp} = 0.16$
and $\sigma_{sim} = 0.19$. The observed average shrink factor is comparable to a circuit error rate
$\xi  \approx -\ln \eta_{exp} \approx 0.62$ which is comparable but smaller than in our numerical simulations
for the Hubbard model in \cref{app:hubbard_noise_models} where the range $1.0 \leq \xi \leq 2.5$ was probed.

The observable expected values shown in \cref{fig:experiment} (left) were estimated the following way. First,
an autocorrelation test (Ljung-Box test \cite{box1970distribution,ljung1978measure}) was performed on the
time-dependent Pauli expected values and the worst-scoring 100 observables were selected, i.e., the ones that are most well described by 
random fluctuations around the mean. This way we can select observables whose expected values are (nearly) the same at each timestep.
We then estimated the time-average of these constant signals. Since 1000 timesteps were recorded with 
200 shots at each timestep, overall $2\times10^4$ shots were available to predict a time-constant expected value
thus the level of shot noise is suppressed to a relatively low level (average standard deviation in \cref{fig:experiment}(left) is
$0.01$, hence no error bars are shown).
  
Second, we use all 220 shots per time step available from the experiment to estimate the shadow spectrum and compare to noisy and ideal
simulations in \cref{fig:experiment}(right). These results clearly confirm our theoretical guarantees that reasonable levels of experimental
noise do not affect the position of the peak. The intensity of the peak is, of course, affected by noise: the experimentally
observed peak is shrunk by a factor $0.52$ which broadly agrees with our estimated $\eta_{exp} \approx 0.54$.
In contrast, the peak of the noisy simulation is shrunk by a factor $0.37$ which is more significant than the estimated
$\eta_{sim} \approx 0.55$ but still within range since the variability across observables was estimated $\sigma_{exp} = 0.19$,
i.e., the observables that give rise to intense peaks are shrunk by a larger amount in the theoretical simulations. 

Additionally, we observe a faint peak at around $E\approx1.2$ in \cref{fig:experiment} (right) in our noise-free simulation.
This faint peak is not visible in the noisy experiment (simulation) as it would require a higher number of shots to resolve --
indeed this peak is analysed in more detail in \cref{app:variationalsimulation} and it is confirmed
that the peak can be resolved when more shots are available (4000 timesteps vs 1000 timesteps) even under gate noise.

\begin{figure*}[t]
	\begin{centering}
		\includegraphics[width=1.0\textwidth]{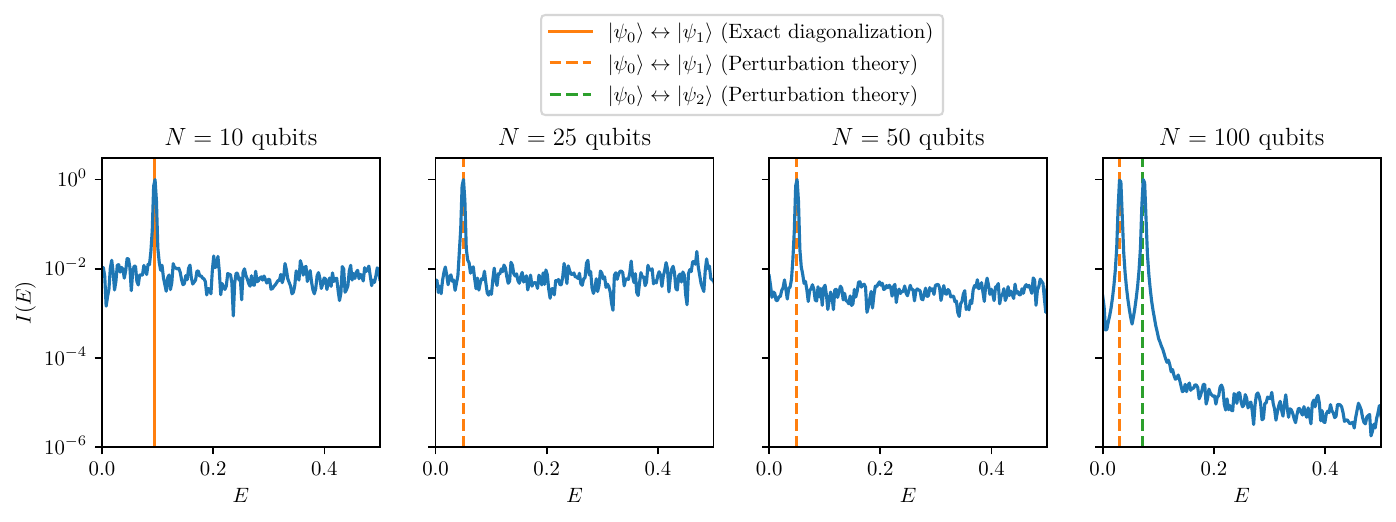}
	\end{centering}
	\caption{\label{fig:mps_spectra}
		Logarithmic intensity shadow spectra of a disordered Heisenberg spin chain ($J=0.01$) up to $N=100$ qubits verifying the scalability of our approach.
			A fixed input state is used, consisting of an equal-weight superposition of the ground state and first excited state of the $J=0$ problem.
			Shadow spectra are computed for system sizes $n\in\{10,25,50,100\}$ using synthetic shot noise equivalent to $N_S=100$ snapshots per timestep.
			At $N=10$ qubits, the spectral peak matches the energy gap between the ground state and first excited state obtained from exact
			diagonalization (vertical solid orange line). At system sizes beyond the reach of exact diagonalization, we compare the peaks to
			second-order perturbation theory (vertical dashed lines), observing peaks corresponding to the gap between the ground and first
			excited states (orange), and at $N=100$ an additional peak corresponding to the gap between ground and second excited states (green).
			Crucially, the standard deviation of the baseline decreases with an increasing number of qubits $N$ (i.e. SNR increases with $N$)
			supporting the scalability of shadow spectroscopy -- at a fixed number of snapshots $N_S$, peaks become \emph{more}
			distinguishable from noise at larger system sizes due to the increased number of available signals.
    }
\end{figure*}

\section{Scaling Analysis~\label{app:scaling}}

\subsection{Probing large system sizes using tensor-network techniques}

Here we simulate large system sizes well beyond the capabilities of exact diagonalization to confirm that our approach is scalable as long as a sufficiently good initial state can be provided -- and to verify that an increasing system size improves the
performance of our approach.
We consider the Heisenberg spin-chain problem with random site-dependent magnetic fields in \cref{eq:spin_ring} at a relatively small coupling $J=0.01$ ($h=1$), so that we can straightforwardly
obtain a good initial state the following way. In particular, in the limit $J \rightarrow 0$ we can analytically solve the
Hamiltonian and simply obtain the ground state as the bitstring $|b\rangle$. 
Here the $j$-th bit $b_j$ of the bitstring $b$ is obtained as
\begin{equation}
    b_j = \begin{cases} 0 & h_j < 0 \\ 1 & h_j \geq 0 \end{cases}
\end{equation}
 The first excited state can then be obtained by flipping the $q$-th bit in this bitstring, where $q$ corresponds to the
absolute smallest on-site energy term $q = \arg\min_k |h_k|$. Thus, an equal superposition of the first and excited
states is simply created by applying a Hadamard gate to the ground state as $H_q |b \rangle$. We note that we have intentionally chosen a system with weak coupling such that (a) a good approximate initial state can be chosen this way, (b) the system accumulates only modest bond dimension when evolving with tensor network methods, remaining tractable and (c) the correctness of resulting peaks can be verified at system sizes too large for exact diagonalization by using perturbation theory.
Given the bond dimension grows exponentially with $JT$, where $T$ is the total simulation time,
choosing a significantly larger coupling  $J$ would get us into a regime where quantum computers can outperform classical
computational techniques.

We simulate shadow spectroscopy by first representing the initial state $H_q |b \rangle$
as a matrix product state (MPS) and evolve it using the time-evolving block decimation (TEBD) algorithm~\cite{vidal2003efficient}.
We perform this evolution using the TEBD implementation in the \textsc{quimb} Python package, using a Trotter step size
\begin{equation}
    \Delta t = \left(\frac{\epsilon_{\text{Trotter}}}{T\|\mathcal{H}\|_F}\right)^{1/k},
\end{equation}
where $\epsilon_{\text{Trotter}}$ is an error tolerance, $T$ is the total evolution time, a $k$-th order Trotter
decomposition is used, and $\|\mathcal{H}\|_F$ is the Frobenius norm of the Hamiltonian
(equivalent to the sum of squares of its Pauli coefficients). 
We set a desired precision $\epsilon_{\text{Trotter}}=10^{-3}$ and use $k=4$. 
Furthermore, at each time step,
we compress the state by performing a singular-value decomposition at each bond, and
discard all contributions with low singular values, i.e., we discard 
singular values below the threshold $\lambda_j / \lambda_0 < \epsilon_{\text{SVD}}$,
where $\epsilon_{\text{SVD}}$ is a cutoff hyperparameter and the singular values are
indexed in descending order. We choose $\epsilon_{\text{SVD}}=10^{-9}$ for our simulations.

This way we obtain a series of time-evolved MPS states as required for shadow spectroscopy, which we then use to compute time-dependent
expectation values for all up-to-3-local Pauli strings through applying appropriate tensor contractions.
We simulate the effects of shot noise by adding Gaussian random numbers of
standard deviation $1/\sqrt{N_S}$ to all observables. We use $N_S=100$ shots per timestep, for which regime
Gaussian noise is already a very good approximation of shot noise in practice.
While at $N=100$ qubits we obtain a large number of time-dependent signals, i.e.,  the number of up to $3$-local Pauli observables
is 4,410,750, our postprocessing protocol still remains tractable on a laptop computer.
We reject signals with $p>0.01$ under a Ljung-Box test, and then perform the postprocessing detailed in the \cref{sec:shadow_spec}
to generate shadow spectra. This procedure is performed for system sizes up to $n=100$ qubits, generating
the spectra depicted in \cref{fig:mps_spectra}.

We compare the observed peaks to results from exact diagonalization (only at $N=10$) and to
second-order Rayleigh-Schr\"{o}dinger perturbation theory ($n\in\{25,50,100\}$).
Apart from a small visible inaccuracy in the green dashed line (where perturbation theory is less accurate for the
higher excited states) we generally find good agreement.
We reiterate that the choice of a perturbative regime was deliberate to allow peak identification, i.e.,
for a weak coupling $J$, the actual values found by shadow spectroscopy are close to the perturbative prediction, but do not exactly agree.
Finally, we confirm that the signal-to-noise ratio of the spectral peaks indeed 
\emph{increase} with the number of qubits $N$ which clearly verifies the scalability of the present approach
and is expected due to our theoretical results.

\subsection{Exact time evolution}
In this section we numerically analyse scaling properties of shadow spectroscopy by using a fixed
input state and performing exact time evolution such that the numerically observed scaling
is independent of the difficulty of state preparation and does not depend on algorithmic errors
or gate errors either -- which we analysed separately.
Here we choose the same spin problem as in our MPS simulations in the previous section (see \cref{eq:spin_ring})  
but use an order of magnitude larger $J$ coupling ($J=0.1$) in order to probe regimes that are more difficult to approximate classically. 
We perform a series of simulations for an increasing number of qubits using an initial state with a fixed overlap with the first
three eigenstates as $|\psi(0)\rangle \propto (1, \tfrac{1}{10}, \tfrac{1}{10}, 0, 0 \dots)$,
and time evolve in $\ttot = 500$ timesteps (using exact time evolution).
We estimate local Pauli observables at each timestep from classical shadows of $1000$ snapshots and
apply the Ljung-Box test to estimate the level of distinguishably of a signal from random noise,
see~\cref{sec:postproc}.

\begin{figure*}
	\centering
	\includegraphics[width=\textwidth]{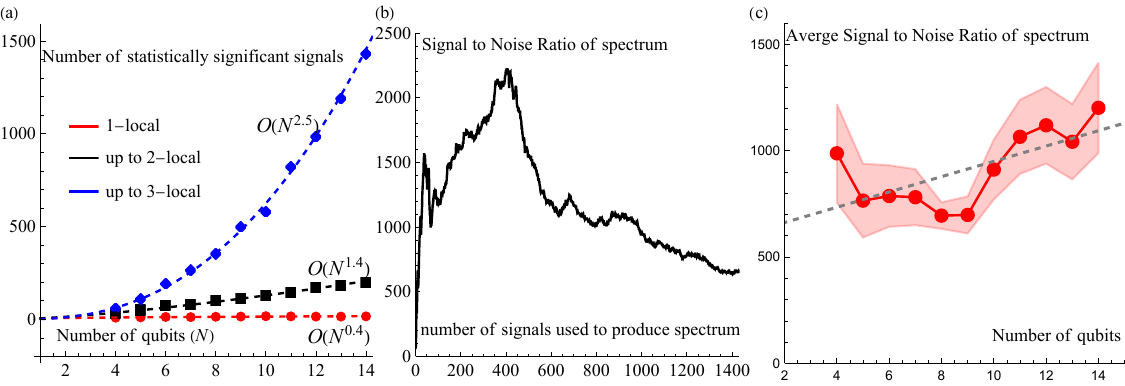}
	\caption{
			Simulations of a spin problem similar to \cref{eq:spin_ring}
			using a fixed input state $|\psi(0)\rangle \propto (1, \tfrac{1}{10}, \tfrac{1}{10}, 0, 0 \dots)$ and
			exact time evolution to analyse how the performance of shadow spectroscopy scales for an increasing number of qubits.
			(a) The number of signals with a Ljung-Box test statistical significance $p <0.01$ using all
			1-local Pauli strings (red) or using up to all 3-local Pauli strings (blue).
			(b) Using an increasing number of signals sorted according to their statistical significance
			at $N=14$ qubits initially significantly increases the spectral SNR but beyond an optimal point
			at $\nobs = 407$ the SNR decreases -- note that in practice one would not mind a decreased 
			SNR as long as the additional signals can resolve new, faint peaks.
			As most of the signals are due to $3$-local Pauli strings, it is indeed beneficial to
			go beyond 1-local or two-local Pauli strings at the expense of only a polynomially increased
			classical post processing.
			(c) The average SNR (red dots) when using all statistically significant signals for an increasing qubit count.
			The slightly increasing SNR confirms that indeed shadow spectroscopy is efficient as long
			as sufficient overlap in the initial state is guaranteed -- which is exponentially hard
			in general for arbitrary Hamiltonians. 
	}
	\label{fig:scaling_analysis}
\end{figure*}

In \cref{fig:scaling}(left) we plot for an increasing number $N$ of qubits the number of statistically significant signals, i.e., 
signals whose Ljung-Box test returns statistical $p$ values satisfying $p<0.01$.
Using only 1-local Pauli observables results in a number of signals in \cref{fig:scaling}(left, red) that only
grows as $\mathcal{O}(N^{0.4})$; This confirms that not all one-local Pauli operators
yield meaningfully intense signals as expected from the entangled nature the spin ring's low-lying eigenstates.
In contrast, we observe that measuring
all up to 3-local Pauli strings yields a significant growth of $\mathcal{O}(N^{2.5})$
of the number of statistically significant signals resulting in a large number $\nobs = 1431$
of intense signals at $N=14$ qubits.
This is indeed advantageous in increasing the Signal-to-Noise-Ratio (SNR) of the spectrum
but can also be used to identify a large number of observables to study physical properties of the relevant
excitations.

In \cref{fig:scaling} (middle) we show the SNR of the shadow spectrum as a function of the
number of observables used to produce the spectrum at $N=14$ qubits by sorting the
signals according to their statistical significance.
To simplify the analysis, we use the post-processing
approach from \cref{app:spectral_density} and simply estimate the Mean Squared Spectrum
from the most intense $\nobs$ signals -- as opposed to randomly selecting observables as in \cref{app:spectral_density}.
This does not require additional quantum resources but only an increased classical post processing time:
While in \cref{sec:postproc} we selected randomly $\nobs$ observables from a pool
and obtained a spectrum with a post processing complexity of $\mathcal{O}(\nobs)$,
here we we perform a Ljung-Box test on every signal in the pool, then sort the signals according to statistical
significance and calculate a plurality of spectra using an increasing number of signals.

We find that the SNR significantly improves as we increase the number of observables $\nobs$
up to an optimal point at $\nobs = 407$. Beyond this point the SNR starts to decrease for an
increasing $\nobs$ as the intensity of the signals with higher $p$ values decreases. However, we note
that typically one does not mind a decreased average SNR as long as the additional
observables beyond the optimal point can resolve  new, faint peaks. 
Furthermore, at the optimal point of $\nobs = 407$, most of the signals are due to 3-local Pauli observables
which demonstrates that it is indeed beneficial to use up to 3-local Pauli observables rather than just using,
e.g., 2-locals or 1-locals.

Finally, in \cref{fig:scaling} (c) we calculate the SNR of the shadow spectrum for an increasing number of qubits.
Specifically, in \cref{fig:scaling} (right, red dots)
an average SNR is plotted which is the mean of the curve in \cref{fig:scaling} (b);
Additionally shown is the standard deviation in \cref{fig:scaling} (c, red shading) which is
simply the standard deviation of the curve in \cref{fig:scaling} (b).
Given the constant overlap with the 
first 3 eigenstates of the spin-problem Hamiltonian---which would be exponentially difficult to guarantee for
general Hamiltonians---the efficacy of shadow spectroscopy even appears to grow slightly  for an increasing qubit
count. This nicely verifies that shadow spectroscopy is indeed efficient given support of the relevant excitation operators 
in the observables estimated from classical shadows and support of the low-lying eigenstates in the initial states.

\section{Noise model analysis \label{app:hubbard_noise_models}}
\begin{figure}[tbhp]
    \centering
    \includegraphics{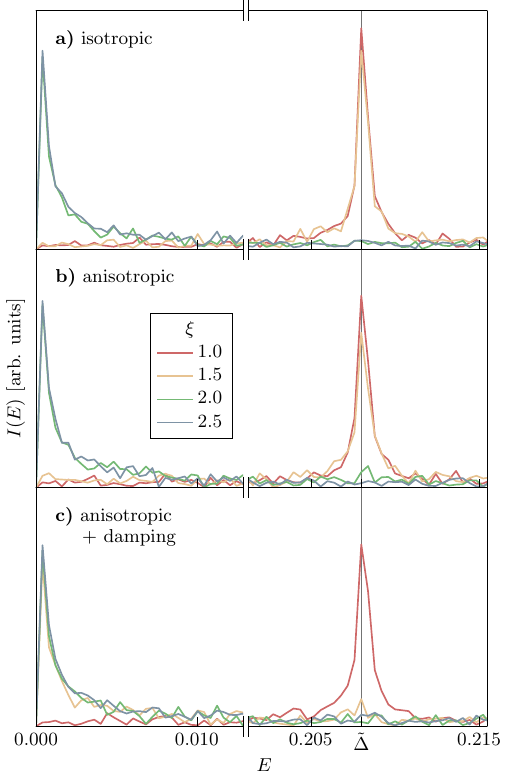}
    \vspace{-1em}
    \caption{
    	Relevant regions of the energy spectra obtained for the Hubbard model using different noise models and various numbers of errors $\xi = \lambda N_\mathrm{gates}$ per full circuit. Notice the interruption of the $E$-axis. The position of the peak without noise using the same Trotter time step is marked as $\tilde{\Delta}$. The results demonstrate robustness to various types of noise. \textbf{a)}~Isotropic depolarising noise, identical to that in Figure~\ref{fig:hubbard_noisy}.
    	\textbf{b)}~Anisotropic depolarising noise with $0.9\,\lambda$ probability of Pauli $Z$ errors,
    	and $0.05\,\lambda$ probability each of a Pauli $X$ or a Pauli $Y$ error.
    	Note that the peak at $\xi = 2$ (green line) seems to be slightly shifted to the right but
    	this is likely just an artefact due to the peak height being comparable to additive shot noise (background noise).
    	\textbf{c)}~Anisotropic depolarising noise like in \emph{b)}, but each depolarising channel is additionally followed by a damping channel with probability $\lambda/10$.
    }
    \label{fig:noise_models}
\end{figure}

We demonstrate our method's robustness to a variety of error models that resemble to dominant noise contributions in
physical gates of near-term quantum devices.
While our theoretical results in \cref{app:noise_rob} generally guarantee robustness under the noise model in \cref{eq:noise_model},
here we numerically probe error channels that go beyond \cref{eq:noise_model}.
We perform a set of calculations on the Fermi-Hubbard model
presented in Section~\ref{sec:hubbardmodel} using the following noise models.
\begin{itemize}[leftmargin=*]
    \item \emph{Isotropic depolarising noise}. Identical to the model used in the main text and described in Appendix~\ref{app:hubbard} whereby
    the depolarising probability is uniform over the Bloch sphere. This model describes a scenario when
    different coherent and incoherent error sources are perfectly averaged out by twirling.
    
    \item \emph{Anisotropic depolarising noise}. In order to more realistically 
    capture noise models of experimental devices that are dominated by 
    T1 and T2 relaxation we consider an anisotropic depolarising model
    whereby  every single-qubit gate is followed by the noise channel of the form
        \begin{multline}
            \hspace{2.5em}\Phi_j^\lambda(\rho) = (1 - \lambda) \rho + \frac{9\,\lambda}{10}\sigma^z_j\,\rho\,\sigma^z_j \\+ \frac{\lambda}{20} \sigma^x_j \, \rho \, \sigma^x_j + \frac{\lambda}{20} \sigma^y_j \, \rho \, \sigma^y_j.
        \end{multline}
     Above, Pauli-$Z$ errors make up $90\,\%$ of all error events, and the rest is split evenly between
     Pauli-$X$ and Pauli-$Y$ errors to reflect that the noise model is dominated by T2 relaxation.
    Multi-qubit Pauli gadgets as in Eq.~\ref{eqn:pauli_gadget} are followed by application of the single-qubit noise channel as
    above on each of the edge qubits $i$ and $j$.
    
    \item \emph{Anisotropic depolarising and damping noise}. 
    We now consider a noise model that violates our theoretical assumptions in \cref{eq:noise_model} and
  	explicitly takes into account the effect of T1 relaxation:
  	this noise model that is identical to the
  	anisotropic depolarising noise described above, but every application of $\Phi_j^\lambda(\rho)$ is followed by a damping channel        \begin{equation}
            \hspace{2.5em}\mathcal{N}_j^\gamma(\rho) = K^{(j)}_0\rho\,{K^{(j)}_0}^\dagger + K^{(j)}_1\rho\,{K^{(j)}_1}^\dagger,
        \end{equation}
        with the usual damping Kraus operators
        \begin{equation*}
            \hspace{2.5em}K^{(j)}_0 = \begin{pmatrix} 1 & 0 \\ 0 & \sqrt{1-\gamma} \end{pmatrix}_j, \quad \text{and}\quad
            K^{(j)}_1 = \begin{pmatrix} 0 & \sqrt{\gamma}\\ 0 & 0 \end{pmatrix}_j,
        \end{equation*}
    that act on qubit $j$. The probability of a damping event is 1/10 that of a depolarising event via $\gamma = \lambda / 10$
    to reflect that the T1 relaxation time is typically longer than the T2 relaxation time.
\end{itemize}

\cref{fig:noise_models} shows the results of the simulations. All three investigated models have the main signal peak at the same position, which matches the theoretical value $\tilde{\Delta}$, indicating that shadow spectroscopy is very robust not only to isotropic and anisotropic depolarising noise -- as we expect from the arguments in \cref{app:noise_rob} -- but even to damping errors, to which \cref{app:noise_rob} does not strictly apply.

Notice that due to the randomness of the classical shadow measurements, \cref{fig:noise_models}a) is slightly different from \cref{fig:hubbard_noisy} in the main text. Most notably, $\xi = 2$ produces a very faint peak in \cref{fig:hubbard_noisy}, but not in \cref{fig:noise_models}a), illustrating that this error rate is right on the edge of the transition from recoverable signal to pure noise. The other, experimentally much more relevant noise levels, however, produce almost identical patterns. Note also that the signal in \cref{fig:noise_models}c) loses intensity much earlier than a) and b) -- already at $\xi = 1.5$ the peak is barely visible. We attribute this to the fact that there is an extra contribution of $\lambda/10$ to the noise in these simulations, which $\xi$ does not account for. But, the important feature of unchanged peak position for the case where noise is present also holds in this case.


%

\end{document}